\documentclass[a4paper,onecolumn,11pt,accepted=2026-11-10]{quantumarticle}
\pdfoutput=1
\usepackage{graphicx}
\usepackage{mathtools,amsmath,amsfonts,amssymb,tikz,amsthm}
\usepackage{hyperref}
\usepackage[numbers]{natbib}
\usepackage{braket}
\usepackage{ascmac}
\usepackage{fancybox}
\usepackage{ulem}
\usepackage{multirow}
\usepackage{cancel}
\usepackage{quantikz}
\allowdisplaybreaks

\theoremstyle{definition}
\newtheorem{theorem}{Theorem}
\newtheorem{definition}{Definition}
\newtheorem{corollary}{Corollary}

\newtheorem{lemma}{Lemma}

\newcommand{\BQP}{\textsf{BQP}}
\newcommand{\NP}{\textup{NP}}
\newcommand{\QMA}{\textsf{QMA}}
\newcommand{\UQMA}{\textsf{UQMA}}
\newcommand{\dUQMA}{\textsf{dUQMA}}
\newcommand{\DUQMA}{\textsf{DUQMA}}

\newcommand{\proofpara}[1]{\par\medskip\noindent{\bfseries #1.}\par\nobreak\smallskip\noindent\ignorespaces}

\title{Computational complexity of Berry phase estimation in topological phases of matter}

\author{Ryu Hayakawa}
\email{ryu.hayakawa@yukawa.kyoto-u.ac.jp}
\affiliation{Yukawa Institute for Theoretical Physics \& The Hakubi Center, Kyoto University, Japan}

\author{Kazuki Sakamoto}
\email{kazuki.sakamoto.osaka@gmail.com}
\affiliation{Graduate School of Engineering Science, The University of Osaka, 1-3 Machikaneyama, Toyonaka, Osaka 560-8531, Japan}

\author{Chusei Kiumi}
\affiliation{Center for Quantum Information and Quantum Biology, The University of Osaka, 1-2 Machikaneyama, Toyonaka 560-0043, Japan}

\begin{document}

\maketitle

\begin{abstract}
The Berry phase is a fundamental quantity in the classification of topological phases of matter. In this paper, we present a new quantum algorithm and several complexity-theoretical results for the Berry phase estimation (BPE) problems. Our new quantum algorithm achieves BPE in a more general setting than previously known quantum algorithms, with a theoretical guarantee. {For the complexity-theoretic results, we consider two main cases.} First, we prove $\mathsf{BQP}$-completeness when we are given a guiding state that has a large overlap with the ground state. This result establishes an exponential quantum speedup for estimating the Berry phase. Second, we prove $\mathsf{UQMA} \cap \mathsf{co}$-$\mathsf{UQMA}$-completeness when we have an $\textit{a priori}$ bound for the ground-state energy. Here, $\mathsf{UQMA}$ is the unique witness version of $\mathsf{QMA}$, and $\mathsf{UQMA} \cap \mathsf{co}$-$\mathsf{UQMA}$ precisely captures the complexity of BPE without the known guiding state. Remarkably, this problem is, to our knowledge, the first natural problem complete for $\mathsf{UQMA} \cap \mathsf{co}$-$\mathsf{UQMA}$. These results advance the role of quantum computing in the study of topological phases of matter and provide a pathway for clarifying the connection between topological phases of matter and computational complexity.
\end{abstract}

{
  \hypersetup{linkcolor=black}
  \tableofcontents
}

\section{Introduction}
In modern condensed matter physics, the classification of topological phases has become a central problem.
Topological invariants, such as Berry phases and Chern numbers, play an important role in this classification~\cite{berry1984quantal,shapere1989geometric,hatsugai2006quantized}.
{For example, in one-dimensional insulators, the Berry phase over the Brillouin zone, known as the Zak phase, can be quantized by symmetries such as inversion or chiral symmetry and can distinguish topologically distinct phases, as in the Su--Schrieffer--Heeger model~\cite{Zak1989,SuSchriefferHeeger1979}.
The Berry phase also gives the bulk electric polarization of one-dimensional crystals, and different quantized polarizations can correspond to distinct symmetry-protected phases~\cite{KingSmithVanderbilt1993,Resta1994}.
In two-dimensional Chern insulators and quantum Hall systems, the Chern number is obtained by integrating the Berry curvature over the Brillouin zone; since the Berry curvature is the local curvature associated with the Berry connection, this provides a higher-dimensional generalization of Berry-phase topology~\cite{ThoulessKohmotoNightingaleNijs1982}.
Similarly, in Thouless pumping, the transported charge over one cycle is described by a Chern number in the combined momentum--time parameter space, again built from Berry curvature~\cite{Thouless1983}. A given Berry phase value can therefore encode physical information such as bulk polarization or the distinction between symmetry-protected phases. Such Berry-phase invariants are typically studied in gapped systems, where the relevant eigenstate or occupied subspace can be followed continuously and the invariant is stable under small perturbations. Beyond topological classification, geometric phases also have applications in quantum information: when the Berry phase is generalized to a degenerate eigenspace, it becomes a non-Abelian holonomy, which underlies holonomic quantum computation~\cite{ZR99}.}
However, computing such topological invariants in quantum many-body systems is generally difficult and remains largely unexplored in strongly correlated large systems.

One potential approach to overcoming the limitations of computing the topological invariants of topological phases of matter is to use quantum computing.
Indeed, there has been growing attention to simulate topological matter with quantum computers and observe topological phase transitions~\cite{roushan2014observation,koh2024realization,mei2020digital,xiao2023robust,smith2022crossing,kitagawa2010exploring}.
Moreover, there are several recent proposals to directly estimate the Berry phase using quantum computers
\cite{murta2020berry,niedermeier2024quantum,tamiya2021calculating,mootz2025efficient}.

However, it is unclear from the previous studies whether there is a {\it quantum advantage} in the application of quantum computing for the classification of the topological phases of matter.
Moreover, previous algorithmic approaches to Berry phase estimation~\cite{murta2020berry} have faced intrinsic limitations that restrict their applicability. This approach relies on a dynamical-phase cancellation trick using forward and backward adiabatic evolutions, which effectively isolates the geometric phase. However, this procedure restricts the estimable Berry phase to the interval $[0,\pi)$ if the Hamiltonian does not have time-reversal symmetry. As a result, the algorithm cannot access the full range of the Berry phase modulo $2\pi$, which limits its utility and generality.

In this paper, in order to rigorously understand the potential and limitations of applying quantum computers for topological phases of matter,
we initiate the study of the {\it computational complexity} of estimating the topological invariants.
As an initial step, we focus on the problem of estimating the Berry phase, which is a central topological invariant in topological phases (Figure~\ref{fig:berryphase_intro}).

\begin{figure}
    \centering
    \includegraphics[width=0.7\linewidth]{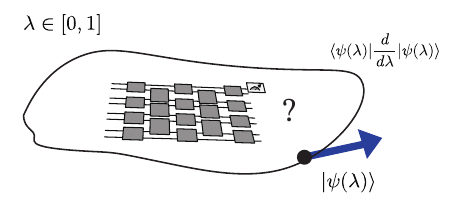}
    \caption{{The problem that we study in this paper: what is the (quantum) computational complexity of estimating the Berry phase? How can we estimate the Berry phase with quantum computers?}}
    \label{fig:berryphase_intro}
\end{figure}

The Berry phase estimation problem is fundamentally different from the local Hamiltonian (LH) problem, which is a canonical problem in quantum computational complexity.
In the LH problem, we are asked to estimate the ground state energy of a given Hamiltonian.
The \QMA-completeness of the LH problem was shown in~\cite{kitaev2002classical} and the result has been extended to various LH problems, including physically motivated Hamiltonians, see ~\cite{cubitt2018universal,piddock2015complexity} for example. {Here, physically motivated Hamiltonians refer to Hamiltonians satisfying structural restrictions commonly arising in many-body physics, such as geometrically local interactions or restricted interaction types.}
The \BQP-completeness of the guided version of the LH problems, i.e., a problem in which we are given an ansatz of the ground state, has also been studied recently~\cite{gharibian2022dequantizing,cade2022improved}.
{An earlier work established \BQP-completeness for a restriction of the LH problem in which a known computational-basis state, or more generally a state with a known polynomial-time construction, has inverse-polynomial overlap with a low-energy state~\cite{richter2007two}.}
{Our Berry phase estimation problem is similar in that a guiding state is used to specify and access a relevant eigenstate, but it is different in the quantity to be estimated. Instead of estimating an eigenenergy, we estimate the geometric phase acquired by the eigenstate under adiabatic evolution along a closed path.} If we adiabatically evolve a (non-degenerate) eigenstate $\ket{\psi}$ of a Hamiltonian according to some closed path, we will obtain the same state. However, even if we obtain the same state, it can acquire a global phase. It is known that after a loop of adiabatic time evolution, the obtained state can be written in the form
$$e^{-i\theta_D}e^{i\theta_B}\ket{\psi}.$$
{This is the usual form of the phase acquired under adiabatic evolution~\cite{berry1984quantal}.}
The first phase $\theta_D$ is called a dynamical phase, and it is dependent on the energy of the evolved state as well as the speed of adiabatic time evolution.
The second phase $\theta_B$ is called the Berry phase or a geometric phase. Interestingly, this term does not depend on the energy and the speed of the time evolution.
{As a consequence, Berry phase estimation targets a geometric property of the eigenstate path, rather than the eigenenergy itself. This contrasts with the usual LH problem, which is an energy-estimation problem.}

\subsection{Our results}

\paragraph{Problem definitions}
We first introduce several setups for the Berry phase estimation problems.
We study the following problem in several settings.\vspace{1em} \\
\noindent
\underline{Berry Phase Estimation (BPE)}
\begin{itemize}
    \item \textbf{Input:} A description of continuously parameterized $k$-local $n$-qubit Hamiltonians $\{H(\lambda)\}_\lambda\in[0,1]$  with $H(0)=H(1)$ that have non-degenerate ground states and inverse-polynomial spectral gap\footnote{{Here we assume that \(H(\lambda)\) is described succinctly as a function of \(\lambda\): given a rational approximation of \(\lambda\) and a requested polynomial precision, the description outputs the local terms and coefficients of \(H(\lambda)\) to that precision in polynomial time. The polynomial bounds on \(\|\partial_\lambda H\|\) and \(\|\partial_\lambda^2 H\|\) in Definition~\ref{def:BPE_guided} ensure that the use of finitely many such descriptions in the algorithm only incurs inverse-polynomial discretization error.}}.
    \item \textbf{Output:} $1$ if $\theta_B \in [a,b]$ and $0$ if $\theta_B \in [c,d]$ {under the promise that either holds}.
\end{itemize}
Here, $[a,b]$ and $[c,d]$ are intervals in $[0,2\pi]$ that are separated by $\delta > 1/\mathrm{poly}(n)$, and $b-a, d-c > 1/\mathrm{poly}(n)$.

We study the complexity of the BPE problem in the following three settings:
\begin{enumerate}
    \item \textbf{BPE with known guiding state} (Definition~\ref{def:BPE_guided}): \\
          In this problem, we are given a succinct description of a \textit{guiding state} for the adiabatic time evolution that has at least inverse-polynomial overlap with the {initial eigenvector} $\ket{\psi(0)}$, which is the ground state of $H(0)$.
    \item \textbf{BPE with energy threshold} (Definition~\ref{def:BPE_withEth}):\\
          {In this problem, we are given a parameter $E_{\mathrm{th}}$ that separates
          the ground-state energy $E_0(0)$ and the first excited-state energy $E_1(0)$ of $H(0)$, satisfying
          \[
              E_0(0) + 1/\mathrm{poly}(n) \leq E_{\mathrm{th}} \leq E_1(0) - 1/\mathrm{poly}(n).
          \]
          }
          The task is to solve the BPE in this situation {\it without} the guiding state.

    \item \textbf{BPE without additional inputs} (Definition~\ref{def:BPE_withoutEth}):\\
          Finally, this is the problem of deciding the BPE without the guiding state and the energy threshold.
\end{enumerate}
For the formal formulation, see the referred definitions.

\paragraph{Our results}
Our results on the computational complexity of these problems are as follows.

\begin{enumerate}
    \item \textbf{BPE with known guiding state}
          is \BQP-complete (Theorem~\ref{thm:BQPcomp}).
    \item \textbf{BPE with energy threshold}
          is {$\mathsf{UQMA} \cap \mathsf{co}$-$\mathsf{UQMA}$}-complete (Theorem~\ref{thm:UQMAcomp}).
    \item \textbf{BPE without energy threshold}
          is
        $\mathsf{P}^{{\mathsf{UQMA} \cap \mathsf{co}-\mathsf{UQMA}}\mathsf{[log]}}$-hard under polynomial-time truth-table reductions and contained in $\mathsf{P}^{\mathsf{PGQMA[log]}}$
({Corollary}~\ref{thm:P^PGQMAcomp}).
\end{enumerate}
{In these hardness results, the Hamiltonian families act on qubits, each local term has operator norm at most \(\mathrm{poly}(n)\), and the hardness holds already for \(5\)-local Hamiltonians.}
{Throughout this paper, hardness and completeness are with respect to polynomial-time many-one reductions (Karp reductions) unless stated otherwise.}

The first main result shows that the BPE with known guiding state (Definition~\ref{def:BPE_guided}) is \BQP-complete (Theorem~\ref{thm:BQPcomp}).
Here, \BQP~is the class of decision problems that can be solved efficiently by quantum computers.
{The containment in BQP is shown through a new quantum algorithm for Berry phase estimation that we introduce in this paper.
The BQP-hardness is shown through {an adaptation of the standard Feynman--Kitaev circuit-to-Hamiltonian construction}.
}
By assuming a standard assumption that $\mathsf{BPP}\neq\BQP$, our result shows {\it a superpolynomial quantum advantage} for this BPE problem.
This result strongly supports the application of quantum computing for studying topological phases of matter with theoretical quantum advantage for universal quantum computers.
{A related open question is stated in Section~\ref{sec:discussion} concerning whether one can show an efficient classical algorithm for the Berry phase estimation in constant precision rather than the inverse-polynomial precision. }

{
Our results for the second and the third problems show that the BPE without guiding states is also captured by certain complexity classes.
}
First, observe that the Berry phase estimation problem for a local Hamiltonian is naturally defined for {\it non-degenerate} local Hamiltonians for the target eigenstate, such as the ground state.
{
Moreover, the BPE problem without a guiding state is unlikely to be in BQP in the worst-case.
This is because preparing ground states of local Hamiltonians is closely related to the local Hamiltonian problem, which is QMA-complete~\cite{kitaev2002classical}.
Thus, without additional structure, quantum computers are not expected to prepare ground state efficiently in general\footnote{
This is a worst-case statement, and thus efficient ground state preparation algorithms may still exist for structured Hamiltonians, such as certain frustration-free Hamiltonians~\cite{schwarz2013information, gilyen2017preparing} and one-dimensional gapped Hamiltonians~\cite{landau2015polynomial}.
}.
}
Therefore, it is natural to expect that the complexity of BPE without the guiding state is relevant to the setting where a prover sends a quantum witness and a verifier performs a polynomial-time quantum computation, such as \QMA.
Among such complexity classes, one relevant complexity class would be
\UQMA~(unique Quantum Merlin-Arthur)
because we are considering a non-degenerate setting.
{
Another important observation is that BPE is closed under taking the complements of the conditions for the output, that is, we are interested in two complement intervals in $[0,2\pi)$.
This implies that the complexity class which captures BPE also has to be closed under taking the complements.
This is again a sharp contrast with the usual local Hamiltonian problems.
As a result, we show that the complexity of BPE without guiding state is captured by $\UQMA \cap \mathsf{co}\text{-}\UQMA$.
}
Interestingly, BPE with energy threshold is, {to the best of our knowledge}, the first natural problem {complete for $\UQMA \cap \mathsf{co}\text{-}\UQMA$}.
{In \cite{aharonov2022pursuit}, it is shown that the problem called unique Local Hamiltonian, in which it is promised that the ground state of the given Hamiltonian is non-degenerate, is complete for the class $\mathsf{UQMA}$.
{An oracle separation between $\mathsf{UQMA}$ and $\mathsf{QMA}$ has recently been shown in \cite{anshu2024uniqueqma}, which can be seen as evidence that $\mathsf{UQMA}$ is strictly weaker than $\mathsf{QMA}$.}
{In this sense, our result identifies BPE with an energy threshold as a natural physical problem in the presumably smaller intersection class \(\UQMA\cap\mathsf{co}\text{-}\UQMA\), in contrast with gapped ground-energy estimation, which is captured by \(\mathsf{PGQMA}\).}
}

\subsection{Proof overview}

\paragraph{Containment and our new quantum algorithm for BPE}
The key technical contribution {in terms of algorithmic aspects of Berry phase estimation} is a new quantum algorithm that removes the symmetry restrictions of Murta et al.~\cite{murta2020berry}, which limited the Berry phase to $[0,\pi)$ {with a simple observation}. Our method compares two adiabatic evolutions with different energy scales: the Berry phase remains invariant, while the dynamical phase rescales, enabling their separation. Using quantum phase estimation, we recover $\theta_B$ modulo $2\pi$ with inverse-polynomial precision.
Given {a polynomial-size classical description of} a guiding state $\ket{c}$ {from which $\ket{c}$ can be efficiently prepared, and assuming that $\ket{c}$ has} inverse-polynomial overlap with $\ket{\psi(0)}$, standard ground-state preparation techniques~\cite{ge2019faster,lin2020near} yield $\ket{\psi(0)}$ efficiently, after which our procedure runs in polynomial depth.
 Hence, BPE with a guiding state is contained in \BQP.
{Containment in \dUQMA~and $\mathsf{P}^{\mathsf{PGQMA[log]}}$} also comes from this new quantum algorithm.

\paragraph{\BQP-hardness}
In order to show the \BQP-hardness of the Berry phase estimation problem, we provide {an adaptation of the standard Feynman--Kitaev circuit-to-Hamiltonian construction}. The constructed Hamiltonian takes the form
$$
    H_r(\lambda) = H_{\mathrm{hist}} + r V(\lambda),
$$
with some choice of $r\in \mathbb{R}_+$.
{The $H_{\mathrm{hist}}$ term comes from a previously known circuit-to-Hamiltonian construction where the ground state of $H_{\mathrm{hist}}$ is given by a so-called history state.
Here we assume that the computation is preceded by \(M\) initial identity gates.
Writing \(L=T+M\) for the total number of gates and \(N=L+1\), we take}
\[
    {
    V(\lambda)=
    \left(
    e^{2\pi i \lambda}\ket{1}\bra{0}_\text{out}
    +e^{-2\pi i \lambda}\ket{0}\bra{1}_\text{out}
    \right)
    \otimes\ket{1}\bra{1}_{C_L}.
    }
\]
{
Thus \(V(\lambda)\) acts on the output qubit only at the final clock state.
The initial-idling step makes the history state have large overlap with an efficiently preparable guiding state, while the final-clock projector extracts the output of the completed computation with amplitude \(1/\sqrt N\).
}
Now we have $H_r(0)=H_r(1)$.
We use a second-order  perturbation theory with a suitable choice of parameters to show that the Berry connection $\mathcal{A}_\lambda=\bra{\psi(\lambda)}\frac{\mathrm{d}}{\mathrm{d}\lambda}\ket{\psi(\lambda)}$ is contained in
$
    \delta \leq i\mathcal{A}_\lambda\leq \frac{\pi}{2}
$
for YES instances of the \BQP-problem and
$
    \delta \leq -i\mathcal{A}_\lambda\leq \frac{\pi}{2}
$ for NO instances.
This will lead to the bounds for intervals of the Berry phase in the YES instances and NO instances.

\paragraph{$\UQMA \cap \mathsf{co}\text{-}\UQMA$-hardness and {$\mathsf{P}^{\UQMA \cap \mathsf{co}\text{-}\UQMA\mathsf{[log]}}$-hardness}}
{
In order to show the $\UQMA \cap \mathsf{co}\text{-}\UQMA$-hardness of the BPE problem without the guiding state,
we apply a reduction analogous to the \BQP-hardness proof.
However, we face a technical issue: we cannot directly use the standard circuit-to-Hamiltonian construction since $\UQMA \cap \mathsf{co}\text{-}\UQMA$ has two verification circuits.
Thus we introduce a new complexity class $\dUQMA$ \footnote{We keep open the position of ``\DUQMA'' for another complexity class, see Section~\ref{sec:discussion}. } (doubly Unique QMA, formally defined in Definition~\ref{dfn:dUQMA}) which has one verification circuit but has two output qubits.
In Section~\ref{sec:dUQMA_characterization}, we show that $\dUQMA$ is equivalent to $\UQMA \cap \mathsf{co}\text{-}\UQMA$ and thus $\dUQMA$-hardness implies $\UQMA \cap \mathsf{co}\text{-}\UQMA$-hardness.
The verification circuit of $\dUQMA$ has the following properties:
\begin{itemize}
    \item \textbf{Yes instances}:  there is a unique witness $\ket{w_{\text{yes}}}$ s.t. the verifier outputs $11$ ($1$ for the first output and $1$ for the second output) with high probability after measuring the two output qubits. For any other witness orthogonal to $\ket{w_{\text{yes}}}$, the probability of outputting 1 after measuring the second output qubit is low.
    \item \textbf{No instances}:
    there is a unique witness $\ket{w_{\text{no}}}$ s.t. the verifier outputs $01$ with high probability after measuring the two output qubits.
    For any other witness orthogonal to $\ket{w_{\text{no}}}$, the probability of outputting 1 after measuring the second output qubit is low.
\end{itemize}
As it can be seen from the above characterization, \dUQMA~possesses ``uniqueness'' {\it for both YES and NO cases}, that's why we have named this new complexity class as ``doubly Unique Quantum Merlin Arthur''.
It is worth explaining why $\dUQMA$ needs two output qubits.
$\UQMA \cap \mathsf{co}\text{-}\UQMA$ have two unique witnesses in total.
Thus the verifier needs to distinguish three cases: (1) unique witness for YES instance, (2) unique witness for NO instance, and (3) states orthogonal to these two states.
This forces the verifier for $\dUQMA$, which is equivalent to $\UQMA \cap \mathsf{co}\text{-}\UQMA$, to have an output space with dimension at least three.
}

{
We also explain the reason that $\dUQMA$ captures the complexity of BPE.
In BPE without the guiding state, we need to verify two circumstances: (1) verification that we are performing adiabatic time evolution for the unique ground state, and (2) the verification of the membership of the obtained Berry phase in the given two intervals.
One of the output qubits corresponds to the verification of the ground state, and the other output qubit corresponds to the Berry phase estimation.
The uniqueness condition in both YES and NO instances, as well as the bounded acceptance probability, allows us to construct Hamiltonians whose unique ground state has an eigenvalue lower than an energy threshold  $E_{\text{th}}$ for {\it both YES and NO cases}.
After a detailed perturbative analysis, we obtain Berry-phase bounds for YES and NO instances that imply the desired \dUQMA-hardness reduction.
The $\mathsf{P}^{{\mathsf{UQMA}\cap\mathsf{co}\text{-}\mathsf{UQMA}}\mathsf{[log]}}$-hardness of the BPE problem without the energy threshold under polynomial-time truth-table reductions can be proven with a similar idea to \cite{yirka2025note}.
}

\subsection{Discussion and conclusion}
\label{sec:discussion}

In this paper, we initiated the study of the computational complexity of estimating a topological invariant that is related to the classification of topological phases of matter.
We have shown a quantum advantage result for this problem by a \BQP-completeness result and complexity results related to variants of unique \QMA.
{The classification of topological phases can be important not only theoretically but also for applications in early fault-tolerant quantum computers because the topological nature makes them robust against local noise~\cite{xiao2023robust}.}
We believe that our results open a wide possibility of studying the task of classification of topological phases of matter from the perspective of quantum computational complexity.
We highlight some open questions below.

\paragraph{Computational complexity of other topological invariants.}
We have studied the computational complexity of estimating the Berry phase.
{Topological phases are characterized not only by Berry phases but also by other topological invariants, such as Chern numbers, winding numbers, and $\mathbb{Z}_2$ indices. A particularly natural next target in our setting is the Chern number. Although the Chern number is formally defined as an integral of the Berry curvature, the Berry curvature can be understood as the infinitesimal Berry phase per unit area in parameter space. Thus, the Chern number can be viewed as a global invariant obtained by accumulating local Berry-phase information over a closed parameter manifold.}
Identifying the computational complexity of computing the Chern number seems to be an interesting open question.

{
\paragraph{Relation to holonomic quantum computation.}
A related but distinct direction is holonomic quantum computation. Zanardi and Rasetti showed that non-Abelian geometric phases associated with adiabatic evolution in degenerate eigenspaces can be used to implement quantum computation \cite{ZR99}. In contrast, the present work focuses on the computational complexity of estimating the Abelian Berry phase of a non-degenerate eigenstate. Thus, our results do not directly address the degenerate, non-Abelian setting underlying holonomic quantum computation. Nevertheless, they suggest a complementary complexity-theoretic question: what is the complexity of estimating non-Abelian holonomies themselves? Progress on this question could help clarify, from a complexity-theoretic viewpoint, how the geometric structures used in holonomic quantum computation relate to computational hardness. We leave this extension to degenerate eigenspaces as an interesting direction for future work.
}

\paragraph{Quantum-inspired classical algorithm for Berry phase estimation}
We have studied the Berry phase estimation problem in the regime of inverse-polynomial precision.
In the energy estimation problem (guided LH problem), classical containment for constant precision problems has been known \cite{gharibian2022dequantizing}.
{More recent classical results for guided LH include a dequantized randomized imaginary-time-evolution algorithm at constant accuracy under specified overlap conditions~\cite{zhang2024dequantized} and, for stoquastic Hamiltonians, a deterministic approximation algorithm under constant promise and spectral gaps, constant overlap, and a guiding state preparable by a constant-depth geometrically local circuit~\cite{waite2025stoquastic}.}
{Those dequantization results apply to energy estimation, where the target quantity can be expressed as a spectral property of a single Hamiltonian relative to a guided state. Extending such techniques to the adiabatic/Berry-phase setting is nontrivial: BPE is a holonomy-type quantity depending on phase-coherent information along a continuous path of Hamiltonians, rather than a spectral estimate for a single Hamiltonian. In particular, one would need to keep track of relative phases between eigenstates at different parameter values.}
An important open question is whether one can show classical containment for the Berry phase estimation in the constant precision regime.

\paragraph{Quantum advantage for Berry phase estimation without the guiding state}
{
Even though we have shown a theoretical superpolynomial quantum advantage for Berry phase estimation, the setting of a given guiding state is rather artificial. Therefore, it is an important open problem to identify a more natural problem for topological phases of matter in which quantum computers gain substantial speedup.
{A related variant is a ``guidable'' version of BPE, where one is promised that a succinctly describable guiding state exists but it is not given as part of the input. In such a formulation, the description of the guiding state could naturally serve as a classical witness, analogously to guidable local Hamiltonian problems~\cite{weggemans2024guidablea}. Determining the precise complexity of this guidable BPE variant is an interesting direction for future work.}
As topological invariants are supposed to be stable against noise, it is also an important open direction to investigate practical quantum advantage for early fault-tolerant quantum computers.
}

\paragraph{Reduction to Physical Hamiltonians}
We have shown our complexity-theoretic hardness for constant local Hamiltonians. However, it is an important open direction to extend our results to more physically motivated Hamiltonians, {by which we mean Hamiltonians with additional structures commonly appearing in condensed-matter models, such as geometric locality in low spatial dimension, nearest-neighbor interactions, or restricted interaction types,} and find more connections between the Hamiltonians considered in condensed matter physics and topological phases.

\paragraph{Related complexity classes}
{We have shown that BPE with an energy threshold is complete for \(\UQMA\cap\mathsf{co}\text{-}\UQMA\). This characterization suggests that Berry phase estimation is naturally tied not only to unique-witness verification, but also to problems whose YES and NO cases both admit unique verifiable witnesses.}

There are other interesting complexity classes as well. First, can one further show the \UQMA-hardness of BPE with energy threshold?  This would lead to $\UQMA = \mathsf{co}\text{-}\UQMA$.

Second, there is a complexity class that can be named to be ``\DUQMA''.
In~\cite{ambainis2014physical}, the complexity class called  \textsf{DQMA} was introduced.
This is a \textsf{QMA}-analog of \textsf{DP}, which is characterized with two languages $L_1\in \mathsf{NP}$ and $L_2\in \mathsf{co}\text{-}\mathsf{NP}$ as $L= L_1 \cap L_2$. (``Difference'' of two {\sf NP} problems.)
Then, we can consider an analogous complexity class \DUQMA~(other than {\sf ``d''UQMA}).
Is there any BPE problem whose complexity corresponds to this \DUQMA?

\subsection{Organization}
The remainder of the paper is organized as follows. In Section~\ref{sec:premliminaries}, we introduce preliminaries. In Section~\ref{sec:main}, we introduce formal definitions of the problems and our main results.
In Section~\ref{sec:improved_algorithm}, we present our new quantum algorithm for Berry phase estimation.
In Section~\ref{sec:dUQMA_characterization}, we give characterizations of the new complexity class \dUQMA.
In Sections~\ref{sec:BQPcomp} and~\ref{sec:UQMA}, we show \BQP-completeness of BPE with guiding state and $\mathsf{UQMA} \cap \mathsf{co}\text{-}\mathsf{UQMA}$-completeness of BPE with energy threshold, respectively. In Section~\ref{sec:P^PGQMAcomp}, we prove Corollary~\ref{thm:P^PGQMAcomp}, establishing $\mathsf{P}^{\mathsf{UQMA} \cap \mathsf{co}\text{-}\mathsf{UQMA}\mathsf{[log]}}$-hardness and containment in $\mathsf{P}^{\mathsf{PGQMA[log]}}$.

\section{Preliminaries}
\label{sec:premliminaries}

\subsection{{Complexity Classes}}
Throughout this paper, we denote by \( \Pi_1 \) the projection onto \( \ket{1} \) in the first qubit—namely, the projector onto the accepting subspace:
\[
    \Pi_1 := \ket{1}\bra{1}_{\text{out}} \otimes I_{\text{other}},
\]
where the first qubit represents the output register of the verifier circuit, and \( I_{\text{other}} \) denotes the identity operator on the remaining qubits.

In this work, we consider promise problems within quantum complexity classes. In the following definitions, $l(\cdot)$ and $m(\cdot)$ denote some polynomials.
One of the fundamental classes is \(\BQP\), which characterizes decision problems solvable efficiently by a quantum computer with bounded error.
\begin{definition}[\BQP]
    A promise problem \(L = (L_{\text{yes}}, L_{\text{no}})\) is in \(\BQP\) if there exists a uniformly generated polynomial-size quantum circuit \(U_x\), acting on \(l(|x|)\) input qubits and \(m(|x|)\) ancilla qubits initialized to \(\ket{0^m}\), such that:
    \begin{itemize}
        \item \textbf{Completeness:} If \(x \in L_{\text{yes}}\), then $\left\| \Pi_1 U_x \ket{0^{l}} \otimes \ket{0^m} \right\|^2 \geq \frac{2}{3},$
        \item \textbf{Soundness:} If \(x \in L_{\text{no}}\), then $\left\| \Pi_1 U_x \ket{0^{l}} \otimes \ket{0^m} \right\|^2 \leq \frac{1}{3}.$
    \end{itemize}
\end{definition}
Another central class is \QMA\ (Quantum Merlin-Arthur), the quantum analogue of the classical class \NP. In \QMA, a quantum verifier receives a quantum witness and decides whether to accept or reject it using a polynomial-time quantum circuit with bounded error.
\begin{definition}[\QMA]
    A promise problem $L = (L_{\mathrm{yes}}, L_{\mathrm{no}})$ is in \QMA\ (Quantum Merlin-Arthur) if there exists a uniformly generated polynomial-size quantum circuit $U_x$, acting on $l(|x|)$ witness qubits and $m(|x|)$ ancilla qubits initialized to $\ket{0^m}$, such that:

    \begin{itemize}
        \item \textbf{Completeness:} If $x \in L_{\mathrm{yes}}$, then there exists a quantum witness $\ket{\psi}$ such that
              \[
                  \left\| \Pi_1 U_x \ket{\psi} \ket{0^m} \right\|^2 \geq \frac{2}{3}.
              \]

        \item \textbf{Soundness:} If $x \in L_{\mathrm{no}}$, then for all quantum states $\ket{\psi}$,
              \[
                  \left\| \Pi_1 U_x \ket{\psi} \ket{0^m} \right\|^2 \leq \frac{1}{3}.
              \]
    \end{itemize}
\end{definition}
Another key complexity class relevant to our work is \UQMA\ (Unique Quantum Merlin-Arthur), a refinement of \QMA\ in which the verifier is promised that there exists a unique accepting quantum witness (up to global phase). This condition aligns with our setting, where the ground state of the Hamiltonian is non-degenerate, and plays an important role in the complexity analysis of Berry phase estimation without a guiding state. Unlike \QMA, which may admit multiple orthogonal accepting witnesses, \UQMA\ requires that only one such witness achieves high acceptance probability.
\begin{definition}[{\UQMA}~\cite{aharonov2022pursuit}]
    A promise problem \(L = (L_{\text{yes}}, L_{\text{no}})\) is in \(\UQMA\) if there exists a uniformly generated polynomial-size quantum circuit \(U_x\), acting on \(l(x)\) input qubits and \(m(x)\) ancilla qubits initialized to \(\ket{0^m}\), such that:
    \begin{itemize}
        \item \textbf{Completeness:} If \(x \in L_{\text{yes}}\), there exists a quantum witness \(\ket{\psi}\) such that
              \[
                  \left\| \Pi_1 U_x \ket{\psi} \otimes \ket{0^m} \right\|^2 \geq \frac{2}{3},
              \]
              and for all orthogonal states \(\ket{\phi} \perp \ket{\psi}\),
              \[
                  \left\| \Pi_1 U_x \ket{\phi} \otimes \ket{0^m} \right\|^2 \leq \frac{1}{3}.
              \]
        \item \textbf{Soundness:} If \(x \in L_{\text{no}}\), then for all states \(\ket{\psi}\),
              \[
                  \left\| \Pi_1 U_x \ket{\psi} \otimes \ket{0^m} \right\|^2 \leq \frac{1}{3}.
              \]
    \end{itemize}
\end{definition}
{
We also provide a formal definition of $\mathsf{UQMA}\cap\mathsf{co}\text{-}\mathsf{UQMA}$
\begin{definition}[$\mathsf{UQMA}\cap\mathsf{co}\text{-}\mathsf{UQMA}$]\label{def:UQMA-co-UQMA}
A promise problem $L = (L_{\text{yes}}, L_{\text{no}})$ is in $\mathsf{UQMA}\cap\mathsf{co}\text{-}\mathsf{UQMA}$ if there exist uniformly generated polynomial-size quantum circuits $U_{\text{yes},x}, U_{\text{no},x}$ acting on $l(x)$ input qubits and $m(x)$ ancilla qubits initialized to $\ket{0^m}$ such that:
\begin{itemize}
    \item If $x\in L_{\mathrm{yes}}$, then
    \begin{align*}
        \exists\ket{\psi_\mathrm{yes}}:\quad
        &\left\| \Pi_1 U_{\text{yes},x} \ket{\psi_\mathrm{yes}} \otimes \ket{0^m} \right\|^2 \geq \frac{2}{3},\\
        \forall \ket{\phi}\perp \ket{\psi_\mathrm{yes}}:\quad
        &\left\| \Pi_1 U_{\text{yes},x} \ket{\phi} \otimes \ket{0^m} \right\|^2 \leq \frac{1}{3},\\
        \forall \ket{\gamma}:\quad
        &\left\| \Pi_1 U_{\text{no},x} \ket{\gamma} \otimes \ket{0^m} \right\|^2 \leq \frac{1}{3}.
    \end{align*}
    \item If $x\in L_{\mathrm{no}}$, then
    \begin{align*}
        \exists\ket{\psi_\mathrm{no}}:\quad
        &\left\| \Pi_1 U_{\text{no},x} \ket{\psi_\mathrm{no}} \otimes \ket{0^m} \right\|^2 \geq \frac{2}{3},\\
        \forall \ket{\phi}\perp \ket{\psi_\mathrm{no}}:\quad
        &\left\| \Pi_1 U_{\text{no},x} \ket{\phi} \otimes \ket{0^m} \right\|^2 \leq \frac{1}{3},\\
        \forall \ket{\gamma}:\quad
        &\left\| \Pi_1 U_{\text{yes},x} \ket{\gamma} \otimes \ket{0^m} \right\|^2 \leq \frac{1}{3}.
    \end{align*}
\end{itemize}
\end{definition}
}

{
Next, we introduce \textsf{PGQMA} (Polynomially-Gapped QMA),
which was introduced in Ref.~\cite{aharonov2022pursuit} to understand \QMA-like problems where not only there is a promise gap between acceptance probabilities of YES and NO instances, but there is also a guaranteed acceptance gap between the witness with the highest acceptance probability and the witness with the second highest probability in both YES and NO cases.
Ref.~\cite{aharonov2022pursuit} showed that the complexity of determining the ground state energy of Hamiltonians with inverse-polynomial spectral gap is captured by \textsf{PGQMA} and that \UQMA\ is equivalent to \textsf{PGQMA} under randomized reductions.
}

{
To introduce \textsf{PGQMA}, we define the accept operator $\mathcal{Q} = (I\otimes \bra{0^m}) U_x^\dagger \Pi_1 U_x (I\otimes \ket{0^m})$,
and let the eigenvalues of $\mathcal{Q}$ be $\lambda_0(\mathcal{Q})\geq \lambda_1(\mathcal{Q}) \geq \dots$.
Notably, the acceptance probability of the witness state is written as the convex combination of these eigenvalues, and QMA can alternatively be defined using the completeness of $\lambda_0(\mathcal{Q})\geq 2/3$ and the soundness of $\lambda_0(\mathcal{Q})\leq 1/3$.
Using this property of the accept operator, we can define \textsf{PGQMA} as follows.
}

\begin{definition}[\textsf{PGQMA} {\cite{aharonov2022pursuit}}]
\label{def:pgqma}
    \leavevmode\par
    A promise problem $L = (L_{\mathrm{yes}}, L_{\mathrm{no}})$ is in \textsf{PGQMA} (Polynomially-Gapped QMA) if there exists a uniformly generated polynomial-size quantum circuit $U_x$, acting on $l(|x|)$ input qubits and $m(|x|)$ ancilla qubits initialized to $\ket{0^m}$, and a polynomially small gap function $\Delta(|x|) \geq 1/\mathrm{poly}(|x|)$ such that:

    \begin{itemize}
        \item \textbf{Completeness:} If $x \in L_{\mathrm{yes}}$, then
        $\lambda_0(\mathcal{Q})\geq 2/3$ and $(\lambda_0(\mathcal{Q}) - \lambda_1(\mathcal{Q})) \geq \Delta(|x|)$.

        \item \textbf{Soundness:} If $x \in L_{\mathrm{no}}$, then
        $\lambda_0(\mathcal{Q})\leq 1/3$ and $(\lambda_0(\mathcal{Q}) - \lambda_1(\mathcal{Q})) \geq \Delta(|x|)$.
    \end{itemize}
\end{definition}

{Ref.~\cite{aharonov2022pursuit} showed that estimating the ground state energy of Hamiltonians promised to have an inverse-polynomial spectral gap is \textsf{PGQMA}-complete.}

\subsection{Berry Phase}
Let $\{H(\boldsymbol{\lambda})\}_{\boldsymbol{\lambda}\in\mathcal{M}}$ be a smooth, $k$-local Hamiltonian family depending on real control parameters $\boldsymbol{\lambda}\in\mathcal{M}\subset\mathbb{R}^{m}$. Assume that the instantaneous ground state \( |\psi(\boldsymbol{\lambda})\rangle \) of $H(\boldsymbol{\lambda})$ is non-degenerate for every \( \boldsymbol{\lambda} \), and let \( \Delta_{\mathrm{min}} \) denote the minimum spectral gap between the ground state and the excited states over the entire parameter space \( \mathcal{M} \). Let $U(T)$ be the adiabatic time-evolution operator
\[U(T)=\mathcal{T} \exp \left(-i \int_0^T H(\boldsymbol{\lambda}(t)) d t\right),\]
where $\mathcal{T}$ denotes time ordering and $t\mapsto\boldsymbol{\lambda}(t)$ parametrises the closed loop
$\mathcal{C}$ with $\boldsymbol{\lambda}(T)=\boldsymbol{\lambda}(0)$. The adiabatic theorem guarantees that the ground state returns to itself up to a phase {with $T\rightarrow \infty$}:

\[
    U(T)
    |\psi(\boldsymbol{\lambda}(0))\rangle
    =
    e^{-i\theta_D}e^{i\theta_B}
    |\psi(\boldsymbol{\lambda}(0))\rangle.
\]
Here $\theta_D\in\mathbb{R}$ is the dynamical phase and $\theta_B\in\mathbb{R}$ is the Berry phase. The dynamical phase is given by
\[
    \theta_D
    =\int_{0}^{T} E_0(\boldsymbol{\lambda}(t))dt
\]
where $E_0(\boldsymbol{\lambda}(t))$ is the instantaneous ground-state energy at time $t$, and the Berry phase is given by
\[
    \theta_B
    = i\int _{0}^{T}\bra{\psi ( \lambda ( t))}\frac{\mathrm{d}}{\mathrm{d}t}\ket{\psi ( \lambda ( t))}\, \mathrm{d}t=  i \oint_{\mathcal{C}}
    \bra{\psi(\boldsymbol{\lambda})}
    \nabla_{\boldsymbol{\lambda}}
    \ket{\psi(\boldsymbol{\lambda})}
    \cdot \mathrm{d}\boldsymbol{\lambda},
\]
where $\nabla_{\!\boldsymbol{\lambda}}$ is the gradient with respect to the control parameters $\boldsymbol{\lambda}$. The first factor, the dynamical phase $\theta_D$, depends explicitly on the instantaneous
ground-state energy $E_0(t)$ and on the total runtime $T$; it therefore changes if we
speed up, slow down, or reverse the evolution.
In contrast, $\theta_B$ is the Berry phase, set solely by the shape of the closed path~$\mathcal{C}$ in parameter space and independent of the rate at which the loop is traversed.

For the remainder of this paper, we specialize in the
single-parameter case \(m=1\).
The control manifold is the circle
\(
S^{1}\cong[0,1)\bigl/(\lambda\sim\lambda+1)
\);
we use the normalized coordinate \(\lambda:=t/T\in[0,1)\) with the periodic
identification \(H(\lambda)=H(\lambda+1)\).
In this setting, the corresponding time-evolution operator becomes
\[
    U(T) = \mathcal{T}\exp\!\left(-i T \int_0^1 H(\lambda)\, \mathrm{d}\lambda\right),
\]
and the Berry phase becomes
\[
    \theta_B =
    \int_{0}^{1}
    i\bra{\psi(\lambda)}\frac{\mathrm{d}}{\mathrm{d}\lambda}\ket{\psi(\lambda)}\,\mathrm{d}\lambda{.}
\]

The time evolution is considered adiabatic if the evolution is sufficiently slow. According to \cite{ambainis2004elementary}, this is ensured when the runtime $T$ satisfies

\begin{equation}\label{eq:adia_condition}
    T\geq \frac{10^5}{\delta_{\mathrm{adia}}^2}\cdot\max \left\{
    \frac{\|\mathrm{d}H/\mathrm{d}\lambda\|^3}{\Delta_{\mathrm{min}}^4}, \frac{\|\mathrm{d}H/\mathrm{d}\lambda\|\cdot \|\mathrm{d}^2H/\mathrm{d}^2\lambda\|}{\Delta^3_{\mathrm{min}}}
    \right\},
\end{equation}
where $\left\|\cdot\right\|$ denotes the standard operator norm and $\delta_{\mathrm{adia}}$ is the allowable adiabatic error i.e., $\|U(T)\ket{\psi(0)}-\ket{\psi(1)}\|\le \delta_{\mathrm{adia}}$. Therefore, if $\left\|\frac{dH}{d\lambda}\right\|$ and $\left\|\frac{d^2H}{d^2\lambda}\right\|$ are bounded by $\mathrm{poly}(n)$ and the minimum spectral gap is bounded as $\Delta_{\mathrm{min}} \ge 1/\mathrm{poly}(n)$, it suffices to take $T = \mathrm{poly}(n,1/\delta_{\mathrm{adia}})$.

\section{Problem definitions and main results}
\label{sec:main}

In this section, we present the problem definitions of the Berry phase estimation and our main results.

In order to define the Berry phase estimation problems, we first introduce notations for the intervals in $[0,2\pi)$.
{For $a,b \in [0, 2\pi)$, let us
introduce the following notation:
$$
    [a,b]_{2\pi}:=
    \begin{cases}
        [a,b]             & \text{ if \ } b\geq a, \\
        [0,b]\cup[a,2\pi) & \text{ if \ } b< a.
    \end{cases}
$$
{Throughout, the arguments of $[\cdot,\cdot]_{2\pi}$ are understood modulo $2\pi$; that is, a real number is identified with its representative in $[0,2\pi)$, so that expressions such as $a-\delta$ and $b+\delta$ are always defined.}
With this definition, {and under the input condition $2\delta<|b-a|<2\pi-2\delta$ imposed in Definitions~\ref{def:BPE_guided}, \ref{def:BPE_withEth} and \ref{def:BPE_withoutEth},} $[a+\delta,b-\delta]_{2\pi}$ and $[b+\delta,a-\delta]_{2\pi}$ are {nonempty} intervals in $[0,2\pi)$ that are separated from each other by $2\delta$. Figure~\ref{fig:intervals} shows such intervals in the case $b>a$.
{Here the upper bound $|b-a|<2\pi-2\delta$ is needed because $[a+\delta,b-\delta]_{2\pi}$ and $[b+\delta,a-\delta]_{2\pi}$ are the two complementary arcs cut at $a$ and $b$, each shrunk by $\delta$ at both ends; their lengths are $|b-a|-2\delta$ and $2\pi-|b-a|-2\delta$ up to the orientation of the pair, so both bounds are required for the two arcs to be nonempty.}
}

\begin{figure}
    \centering
    \includegraphics[width=0.5\linewidth]{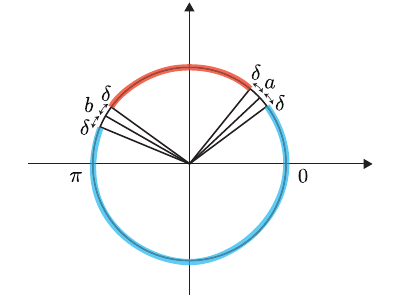}
    \caption{The Berry phase estimation problem is to decide whether $\theta_B$ is in the red region or the blue region in $[0,2\pi)$ for the given $a,b$.}
    \label{fig:intervals}
\end{figure}

Then, the Berry phase estimation problems can be formally defined as follows.
First, we define a problem in which an approximation of the ground state is known as an input of the problem.

\begin{definition}[Berry phase estimation with a guiding state]
    \label{def:BPE_guided}
    \ \\
    \textbf{Input:}
    A family of $k$-local Hamiltonians $\{H(\lambda)\}_\lambda$ parameterized by $\lambda\in [0,1]$ on $n$-qubit {given as a $poly(n)$-size description of a function of $\lambda$}, real numbers $\delta,
        \Delta_{\mathrm{min}}>1/\mathrm{poly}(n)$,
    {$1/\mathrm{poly}(n)\leq \gamma \leq  1-1/\mathrm{poly}(n)$,}
    and
    {$a,b\in [0,2\pi)$}, and a $\mathrm{poly}(n)$-size description of a normalized quantum state $\ket{c}$ s.t. {$2\delta<|b-a|<2\pi-2\delta$}.\\
    {\textbf{Promise: }}
    \begin{itemize}
        \item {($\mathcal{C}1$):} $H(0)=H(1)$, and for all $\lambda\in [0,1]$ the following holds: (1) $\|H(\lambda)\|\in \mathrm{poly}(n)$,  (2) $H(\lambda)$ has a non-degenerate ground state, (3) the spectral gap of $\{H(\lambda)\}_\lambda$ is lower bounded by $\Delta_{\mathrm{min}}$, (4) $\|\mathrm{d}H/\mathrm{d}\lambda\|,\|\mathrm{d}^2H/\mathrm{d}\lambda^2\|\in \mathrm{poly}(n)$.
        \item {($\mathcal{C}2$):}
              {$|\braket{\psi(0)|c}|^2\geq \gamma$}, where $\ket{\psi(0)}$ is a normalized ground state of $H(0)$.
        \item
              The Berry phase $\theta_B =
                  \int_{0}^{1}
                  i\bra{\psi(\lambda)}\frac{\mathrm{d}}{\mathrm{d}\lambda}\ket{\psi(\lambda)}
                  \mathrm{d}\lambda$  is either in  $[a+\delta,b-\delta]_{2\pi}$ or $[b+\delta,a-\delta]_{2\pi}$.
    \end{itemize}
    \textbf{Output:}
    {1 if $\theta_B\in[a+\delta,b-\delta]_{2\pi}$ and output 0 if $\theta_B\in[b+\delta,a-\delta]_{2\pi}$.}
\end{definition}
{Following the standard guided and guidable local Hamiltonian settings~\cite{gharibian2022dequantizing,cade2022improved,weggemans2024guidablea}, the polynomial-size classical description of the guiding state \(\ket{c}\) is understood as a description from which a polynomial-size quantum circuit coherently preparing \(\ket{c}\) can be generated in classical polynomial time.}
{
Here, the guiding state is used to specify and access the relevant initial ground state of \(H(0)\) through the overlap condition ($\mathcal{C}$2). In this sense, it plays the role of a trial or reference state that has non-negligible support on the eigenstate whose Berry phase is to be estimated. The condition ($\mathcal{C}$2) concerns only the magnitude of the inner product with the ground state, and does not impose a prescribed absolute phase relation between the guiding state and the eigenstate. This is consistent with the fact that the phase convention of an eigenstate is a gauge choice, while the Berry phase associated with a closed loop is gauge invariant modulo \(2\pi\). Thus, the guiding state is not assumed to encode the Berry phase itself; it is used only to single out the eigenstate whose geometric phase is studied.
}

\begin{theorem}\label{thm:BQPcomp}
    Berry phase estimation problem with a known guiding state is \BQP-complete.
    {The hardness holds even when restricted to \(5\)-local Hamiltonian families on qubits with polynomially bounded interaction strengths.}
\end{theorem}

The full proof of Theorem~\ref{thm:BQPcomp} is provided in Section~\ref{sec:BQPcomp}, which is divided into two parts: containment in \BQP\ and \BQP-hardness.
{
We remark that the guiding state used in our hardness construction,
defined explicitly in Eq.~\eqref{eq:guiding} of Section~\ref{sec:BQP_hardness}
, is a uniform superposition over a polynomial-size set of computational-basis states and hence is a semi-classical subset state of the type studied in guided and guidable local Hamiltonian problems~\cite{gharibian2022dequantizing,cade2022improved,weggemans2024guidablea}.
It arises from the standard initial-idling mechanism used in such hardness constructions~\cite{cade2022improved,weggemans2024guidablea}.
More recent work has extended guided-LH hardness to broader physically motivated families, including fixed-weight states, matrix product states, Gaussian states, and Fendley states~\cite{waite2025physically}
.
However, since our goal is to establish the hardness of Berry phase estimation, we do not pursue further extensions of the reduction to other guiding-state families.
Thus, the hardness result holds within this previously studied family of guiding states, with $\gamma = 1 - \Omega(1/\mathrm{poly}(n))$; the new ingredient here is the Hamiltonian loop and its final-clock perturbation used to encode the computational output into the Berry phase.
}
In Section~\ref{sec:BQP_containment}, we show that the Berry phase estimation problem---when a suitable guiding state is provided---can be solved efficiently by a quantum computer. That is, the problem lies in the complexity class \textsf{BQP}. Our argument builds on the algorithm introduced in Section~\ref{sec:improved_algorithm}, which estimates the Berry phase \( \theta_B \in [0, 2\pi) \).
Our result on the Berry phase estimation algorithm can be stated as follows.

\begin{theorem}[Quantum algorithm for Berry phase estimation]
    \label{thm:BPE_algorithm}
    Under conditions (\(\mathcal{C}1\))-(\(\mathcal{C}2\)) of Definition~\ref{def:BPE_guided}, for any additive accuracy
    $\varepsilon_B \ge 1/\mathrm{poly}(n)$ and failure probability $\eta \in (0,1/3)$,
    there exists a quantum algorithm that outputs an estimate $\hat\theta$ of the Berry phase
    $\theta_B$ satisfying $|\hat\theta-\theta_B|\le \varepsilon_B \pmod{2\pi}$ and runs in $\mathcal{O}(\mathrm{poly}\!\left(n\right)\cdot \frac{1}{\varepsilon_B}\log{(1/\eta))}$-time.
\end{theorem}

Our algorithm overcomes the limitations of the previous quantum algorithm for the Berry phase estimation of~\cite{murta2020berry}, which is reviewed in Appendix~\ref{sec:berry_algorithms}, because our algorithm can estimate the Berry phase without requiring time-reversal symmetry or limiting the phase to a specific range.

In Section~\ref{sec:BQP_hardness}, we prove the reverse direction: that estimating the Berry phase with a known guiding state is at least as hard as solving any problem in \textsf{BQP}. To do this, we construct a family of quantum systems where the Berry phase depends sharply on the outcome of a quantum computation. The phase differs significantly between YES and NO instances, so by estimating it, one can decide the original problem. This establishes the \textsf{BQP}-hardness of the given problem.

Next, we introduce a Berry phase estimation problem in which an initial ground state is {\it unknown}.
In this problem, we first consider a setting where we are given an energy threshold that separates the ground state energy and the first excited state energy for the initial Hamiltonian.
{Such a threshold is natural spectral prior information in Hamiltonian problems, analogous to the energy thresholds appearing in local Hamiltonian formulations. In the present BPE setting, it is useful because it allows a verifier to test whether a proposed witness lies in the relevant ground-state subspace of \(H(0)\), before estimating the Berry phase associated with that eigenstate.}

\begin{definition}[Berry phase estimation with energy threshold]
    \label{def:BPE_withEth}
    \ \\
    \textbf{Input:}
    A family of $k$-local Hamiltonians $\{H(\lambda)\}_\lambda$ parameterized by $\lambda\in [0,1]$ on $n$-qubit and real numbers $a,b\in [0,2\pi)$, $\delta,
        \Delta_{\mathrm{min}}>1/\mathrm{poly}(n)$,
    $E_{\text{th}}$
    s.t. {$2\delta<|b-a|<2\pi-2\delta$}.\\
    {\textbf{Promise: }}
    \begin{itemize}
        \item {($\mathcal{C}1$) of Definition~\ref{def:BPE_guided} holds.}
        \item  {The energy threshold $E_{\text{th}}$ satisfies $E_0 + \Delta_{\mathrm{min}}/4 \leq E_{\text{th}} \leq E_1 - \Delta_{\mathrm{min}}/4$, where $E_0$ and $E_1$ are the ground state energy and the first excited state energy of $H(0)$.}
        \item The Berry phase $\theta_B =
        \int_{0}^{1}
        i\bra{\psi(\lambda)}\frac{\mathrm{d}}{\mathrm{d}\lambda}\ket{\psi(\lambda)}
        \mathrm{d}\lambda$  is either in  $[a+\delta,b-\delta]_{2\pi}$ or $[b+\delta,a-\delta]_{2\pi}$.
    \end{itemize}
    \textbf{Output:}
    {1 if $\theta_B\in[a+\delta,b-\delta]_{2\pi}$ and output 0 if $\theta_B\in[b+\delta,a-\delta]_{2\pi}$.}
\end{definition}

\noindent
The following theorem completely characterizes the computational complexity of the Berry phase estimation problem with an energy threshold.
\begin{theorem}
    \label{thm:UQMAcomp}
    Berry phase estimation with energy threshold is {$\UQMA \cap \mathsf{co}\text{-}\UQMA$-complete}.
    {The hardness construction again uses \(5\)-local Hamiltonian families on qubits with polynomially bounded interaction strengths.}
\end{theorem}
\noindent
{
In Section~\ref{sec:UQMA_containment}, we prove the containment in $\UQMA \cap \mathsf{co}\text{-}\UQMA$.
In Section~\ref{sec:UQMA_hardness}, we prove the $\UQMA \cap \mathsf{co}\text{-}\UQMA$-hardness.
To show the hardness, we introduce another representation of $\UQMA \cap \mathsf{co}\text{-}\UQMA$ with a single verifier, which we call $\dUQMA$.
We note that, remarkably, BPE with energy threshold is the first natural problem complete for $\UQMA \cap \mathsf{co}\text{-}\UQMA$.
}

Finally, we consider the general problem of estimating the Berry phase without an energy threshold.

\begin{definition}[Berry phase estimation]
    \label{def:BPE_withoutEth}
    \ \\
    \textbf{Input:}
    A family of $k$-local Hamiltonians $\{H(\lambda)\}_\lambda$ parameterized by $\lambda\in [0,1]$ on $n$-qubit, and real numbers $a,b\in [0,2\pi)$, $\delta,
        \Delta_{\mathrm{min}}>1/\mathrm{poly}(n)$
    s.t. {$2\delta<|b-a|<2\pi-2\delta$}, and {($\mathcal{C}1$) of Definition~\ref{def:BPE_guided} holds}. \\
    \textbf{Promise: }
    $\theta_B =
        \int_{0}^{1}
        i\bra{\psi(\lambda)}\frac{\mathrm{d}}{\mathrm{d}\lambda}\ket{\psi(\lambda)}
        \mathrm{d}\lambda$  is either in  $[a+\delta,b-\delta]_{2\pi}$ or $[b+\delta,a-\delta]_{2\pi}$.\\
    \textbf{Output:}
    output 1 if ${\theta_B} \in [a+\delta,b-\delta]_{2\pi}$ and output 0 if ${\theta_B} \in [b+\delta,a-\delta]_{2\pi}$.
\end{definition}
\begin{corollary}
    \label{thm:P^PGQMAcomp}
    Berry phase estimation is contained in $\mathsf{P}^{\mathsf{PGQMA[log]}}$ and $\mathsf{P}^{{\mathsf{UQMA}\cap\mathsf{co}\text{-}\mathsf{UQMA}}\mathsf{[log]}}$-{hard}
    under polynomial-time truth-table reductions.
    {The hardness part uses the same \(5\)-local qubit Hamiltonian construction as Theorem~\ref{thm:UQMAcomp}.}
\end{corollary}
\noindent
In Section~\ref{sec:P^PGQMA_containment}, we prove the containment in $\mathsf{P}^{\mathsf{PGQMA[log]}}$. In Section~\ref{sec:P^dUQMA_hardness}, we prove the $\mathsf{P}^{{\mathsf{UQMA}\cap\mathsf{co}\text{-}\mathsf{UQMA}}\mathsf{[log]}}$-hardness.

\section{New Quantum Algorithm for Berry phase Estimation\label{sec:improved_algorithm}}

In this section, we present an improved quantum algorithm for estimating the Berry phase \( \theta_B \in [0, 2\pi) \) without relying on time-reversal symmetry. This addresses a key limitation of prior work—such as the algorithm by Murta et al.~\cite{murta2020berry}—which restricts the estimable range of \( \theta_B \) to \( [0, \pi) \) due to a phase-doubling technique.
For the convenience of readers, we give an overview of the algorithm in Ref.~\cite{murta2020berry} in Appendix~\ref{sec:berry_algorithms}.

Our central idea is to compare two adiabatic evolutions of the same closed Hamiltonian path,
but executed with different total runtimes.
Increasing the runtime by a factor $\alpha>1$ leaves the geometric (Berry) phase $\theta_B$ unchanged,
while it multiplies the dynamical phase $\theta_D$ by $\alpha$.
By running quantum phase estimation on the adiabatic evolution for both runtime choices and combining the
measurement outcomes, we can cancel the dynamical contribution and recover the Berry phase.
This runtime-scaling technique enables the extraction of $\theta_B$ modulo $2\pi$ over the full interval $[0,2\pi)$, without relying on symmetry assumptions such as time-reversal symmetry.
            \begin{figure}[h]
                \centering
                \begin{quantikz}[row sep=0.6cm, column sep=0.35cm]
                    \lstick[wires=1]{$\ket{+^m}$}  & \ctrl{1} & \ctrl{1} & \qw & \ctrl{1} & \gate{QFT^{-1}} & \meter{} &  \rstick{$\theta_B {-}\theta_D$} \\
                    \lstick{$\ket{\psi(0)}$}  & \gate{U(T)} & \gate{U(T)^2} & \dots & \gate{U(T)^{2^{m-1}}} & \qw & \qw & \qw\\[1.0cm]
                    \lstick[wires=1]{$\ket{+^m}$} & \ctrl{1} & \ctrl{1} & \qw & \ctrl{1} & \gate{QFT^{-1}} & \meter{}  & \rstick{$\theta_B {-} \alpha\theta_D$}\\
                    \lstick{$\ket{\psi(0)}$}  & \gate{U(\alpha T)} & \gate{U(\alpha T)^2} & \dots & \gate{U(\alpha T)^{2^{m-1}}} & \qw & \qw  & \qw
                \end{quantikz}
                \caption{
                    Schematic of the Berry phase estimation algorithm.
                    Two independent quantum phase estimation (QPE) procedures are run on the adiabatic propagators
                    $U(T)$ and $U(\alpha T)$, each with input $\ket{\psi(0)}$ prepared from the guiding state $\ket{c}$.
                    The outcomes $m_T$ and $m_{\alpha T}$ encode the phases
                    $\varphi_1 \equiv \theta_B {-} \theta_D$ and
                    $\varphi_\alpha \equiv \theta_B {-} \alpha\theta_D \pmod{2\pi}$,
                    respectively. These classical outputs are then combined in post-processing to cancel the dynamical phase and reconstruct the Berry phase $\theta_B$.}
                \label{fig:twoQPE}
            \end{figure}
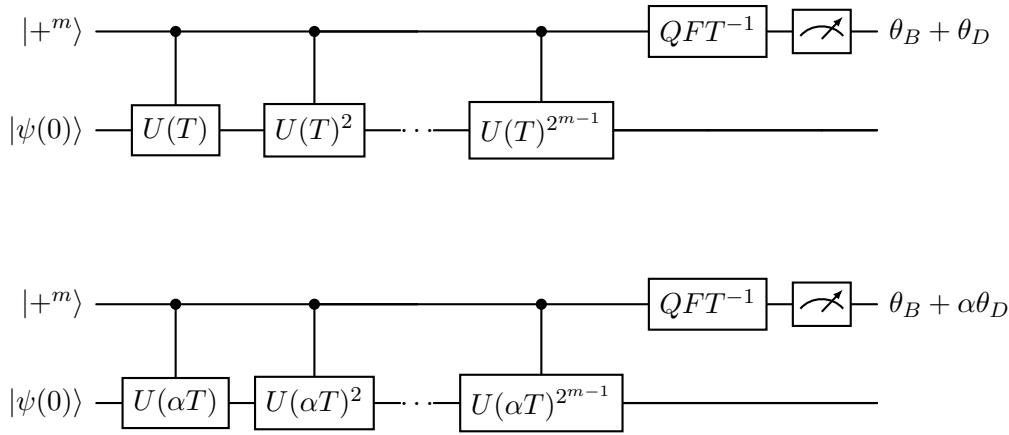
            \\[10pt]
    \noindent\textbf{Algorithm Description}
    \\
    The algorithm is defined under the assumptions of Definition~\ref{def:BPE_guided}.

    For any target additive accuracy $\varepsilon_B \ge 1/\mathrm{poly}(n)$ and target
    failure probability $\eta \in (0,1/3)$, we assume that time evolutions of local
    Hamiltonians, such as $e^{-iHt}$ and their adiabatic counterparts, can be implemented
    efficiently; in the \textsf{BQP} setting, this is without loss of generality, since
    Hamiltonian simulation achieves exponentially precise approximations at polynomial cost,
    well within the accuracy requirement $\varepsilon_B$.

    We also assume access to an efficient ground-state preparation primitive that maps the guiding
    state $\ket{c}$ to the ground state $\ket{\psi(0)}$ of $H(0)$ with constant success probability.
    This is without loss of generality under our gap-and-overlap promises: by the near-optimal method of
    Lin and Tong~\cite{lin2020near} (polynomial filtering via QSP with block-encodings and a binary-search
    over energies), one can prepare $\ket{\psi(0)}$ from $\ket{c}$ in $\mathrm{poly}(n)$ time even without a
    priori knowledge of $E_0$; The preparation error and the failure probability can be suppressed exponentially with polynomial overhead. Hence, for clarity of exposition, we treat the prepared input as exactly $\ket{\psi(0)}$ in the analysis that follows.

    {
Under the gap-and-overlap promises, the initial ground state \(|\psi(0)\rangle\) can be prepared efficiently from the guiding state \(|c\rangle\).
Indeed, using the ground-state preparation method of Lin and Tong~\cite{lin2020near}, implemented by polynomial filtering with a block-encoding of \(H(0)\), one can prepare a state that is inverse-polynomially close to \(|\psi(0)\rangle\) in polynomial time.
This procedure does not require prior knowledge of the exact ground-state energy; the relevant energy window can be located to sufficient precision using the binary-search procedure of Ref.~\cite{lin2020near}.
The preparation error and failure probability can be made inverse-polynomially small, or even exponentially small, with polynomial overhead.
Thus, this preparation step is part of the BQP algorithm and does not constitute an additional oracle assumption.
Possible global phases arising from separate ground-state preparation or QPE runs do not affect the algorithm, since such phases are physically unobservable and the Berry phase around a closed loop is gauge invariant modulo \(2\pi\).
For notational simplicity, we treat the prepared initial state as exactly \(|\psi(0)\rangle\) in the analysis below.
}

    The output of the algorithm is an estimate $\hat\theta$ of the Berry phase $\theta_B$,
    guaranteed with an overall success probability of at least $1-\eta$ to satisfy
    \[
        |\hat\theta - \theta_B| \;\le\; \varepsilon_B \pmod{2\pi}.
    \]
\begin{enumerate}
        \item \textbf{Phase readout with runtime scaling parameters.}
              \\
             Fix the QPE failure parameter $\eta_{\mathrm{QPE}} \le 1-(1-\eta)^{1/4}$ and choose the adiabatic accuracy so that
\[
{\delta_{\mathrm{adia}} \le \frac{{\eta_{\mathrm{QPE}}(1-\eta_{\mathrm{QPE}})}}{2^t-1}.}
\]
{Here $t$ is the number of precision qubits used in the QPE procedure below. This stronger accuracy condition accounts for the accumulation of adiabatic error over all controlled powers used in QPE. For each QPE procedure, these controlled powers use at most $1+2+\cdots+2^{t-1}=2^t-1$ applications of the corresponding adiabatic unitary. Hence, by the triangle inequality, the accumulated adiabatic error{, i.e., the distance $\|\cdot\|$ between the actual final state of the QPE circuit and the ideal one in which every adiabatic unitary acts on $\ket{\psi(0)}$ exactly,} is bounded by}
\[
{\Delta_{\mathrm{adia}} := (2^t-1)\delta_{\mathrm{adia}} \le {\eta_{\mathrm{QPE}}(1-\eta_{\mathrm{QPE}})}.}
\] {Since $t=\Theta(\log(1/\varepsilon_{\mathrm{ph}}))$ and $\varepsilon_{\mathrm{ph}}=1/\mathrm{poly}(n)$, we have $2^t=\mathrm{poly}(n)$, and hence this stronger adiabatic accuracy requirement remains inverse-polynomial.} Let $H_{\max}=\max_{s\in[0,1]}\|H(s)\|$; then $|\theta_D|\le T H_{\max}$ (since $|E_0(t)|\le\|H(t)\|$). Choose $\alpha$ either with $1/(\alpha-1)\in\mathbb{Z}_{>0}$, or to satisfy
              \begin{equation}\label{eq:alpha_choice}
                  \alpha = 1 + \frac{\pi}{T H_{\max} + 2 \varepsilon_B} > 1,
              \end{equation}
              which implies Eq.~(\ref{eq:unwrap_condition}) and ensures that the dynamical phase can be unwrapped correctly even if $1/(\alpha-1)$ is not an integer.

              Define $\varphi_1\equiv\theta_B{-}\theta_D \pmod{2\pi}$ and $\varphi_\alpha\equiv\theta_B{-}\alpha\theta_D \pmod{2\pi}$. Run QPE on $U(T)$ and on $U(\alpha T)$ with additive phase precision
              \[
                  \varepsilon_{\mathrm{ph}}=\frac{\alpha-1}{\alpha+1}\,\varepsilon_B = 1/\mathrm{poly}(n),
              \]
              using $t=\Theta(\log(1/\varepsilon_{\mathrm{ph}}))$ precision qubits. {Reduce the failure probability of each QPE estimate to at most $\eta_{\mathrm{QPE}}$, either by using $R=\Theta(\log(1/\eta_{\mathrm{QPE}}))$ independent repetitions with a circular median or by amplitude amplification.} Then, {since $\bigl|\bra{\phi}P\ket{\phi}-\bra{\psi}P\ket{\psi}\bigr|\le\|\ket{\phi}-\ket{\psi}\|$ for any pure states $\ket{\phi},\ket{\psi}$ and projector $P$,} the outcomes $m_1, m_\alpha$ satisfy
             \[
    \Pr\big[\,|m_{1}-\varphi_1|\le \varepsilon_{\mathrm{ph}}\,\big]\ \ge\ {(1-\eta_{\mathrm{QPE}})-\Delta_{\mathrm{adia}}}\ \ge\ (1-\eta_{\mathrm{QPE}})^2\ \ge\ \sqrt{1-\eta},
\]
              and likewise for $m_\alpha$.

              We re-prepare $\ket{\psi(0)}$ and use independent QPE randomness for the two estimates, so the two success events are independent. Therefore,
              \[
                  \Pr[\text{both succeed}]
                  \ \ge\ 1-\eta.
              \]
              The per-phase cost is
              \[
                  \operatorname{poly}(n)\cdot \Theta\!\left(\frac{T}{\varepsilon_{\mathrm{ph}}}\,\log\frac{1}{\eta_{\mathrm{QPE}}}\right)
                  =\operatorname{poly}(n)\cdot \Theta\!\left(\frac{T}{\varepsilon_{B}}\,\log\frac{1}{\eta}\right),
              \]
              so both phases are obtained to error $\varepsilon_{\mathrm{ph}}$ with total runtime
              $\operatorname{poly}(n,\,1/\varepsilon_B,\,\log(1/\eta))$.
        \item \textbf{Algebraic reconstruction.}
              \\
              For any interval $I=[a,a+2\pi)$ of length $2\pi$, define
              $
                  (x)_{I} := x - 2\pi \cdot
                  \Big\lfloor \frac{x-a}{2\pi} \Big\rfloor .
              $
              Then $(x)_{I}\in I$ and
              $(x)_{I}\equiv x \pmod{2\pi}$.
              With probability $\ge 1-\eta$, both $m_{1}$ and $m_{\alpha}$ lie within
              $\varepsilon_{\mathrm{ph}}$ of their true phases $\varphi_1,\varphi_\alpha$.
              On this event,
              \[
                  \big|\big(m_{1}{-}m_{\alpha}\big)_{(-\pi,\pi]} - ((\alpha-1)\theta_D)_{(-\pi,\pi]}\big|\le 2\varepsilon_{\mathrm{ph}}.
              \]
              Define the estimators
              \[
                  \widehat{\theta}_D := \left(\frac{\big(m_{1}{-}m_{\alpha}\big)_{(-\pi,\pi]}}{\alpha-1}\right)_{[0,2\pi)},
                  \qquad
                  \widehat{\theta}_B := \big(m_{1}{+}\widehat{\theta}_D\big)_{[0,2\pi)}.
              \]
              Since $((\alpha-1)\theta_D)_{(-\pi,\pi]}\equiv (\alpha-1)\theta_D \pmod{2\pi}$, $\widehat{\theta}_D$ will correctly estimate the phase $\theta_D$ after rescaling if $\frac{1}{\alpha-1}$ is an integer. If we choose $\alpha$ such that $\tfrac{1}{\alpha - 1}$ is not integer, which mean satisfying (\ref{eq:alpha_choice}), then we get the following bound instead:
              \begin{equation}\label{eq:unwrap_condition}
                  |\theta_D(\alpha-1)| \le T H_{\max}\cdot \frac{\pi}{TH_{\max}+2\varepsilon_{B}} < \pi - 2\varepsilon_{\mathrm{ph}},
              \end{equation}
              allowing us to estimate the $\theta_D$ unambiguously after multiplying by $\frac{1}{\alpha-1}$. Then, on the same event,
              \[
                  \big|\widehat{\theta}_D - \theta_D\big|
                  \le \frac{2\varepsilon_{\mathrm{ph}}}{\alpha-1},
                  \qquad
                  \big|\widehat{\theta}_B - \theta_B\big|
                  \le \varepsilon_{\mathrm{ph}} + \frac{2\varepsilon_{\mathrm{ph}}}{\alpha-1}.
              \]
              Since $\varepsilon_{\mathrm{ph}}=\frac{\alpha-1}{\alpha+1}\varepsilon_B$, we obtain
              $|\widehat{\theta}_B-\theta_B|\le \varepsilon_B$. Repeating the procedure $O(\log(1/\eta))$ times and taking the circular median yields a success probability $\ge 1-\eta$.
        \item \textbf{Output.}
              Return $\widehat{\theta}_B$.
    \end{enumerate}

    \section{{$\UQMA \cap \mathsf{co}\text{-}\UQMA$ with a single verifier}}
    \label{sec:dUQMA_characterization}
    {
    In this section, we introduce a new complexity class $\dUQMA$ which has a single verifier.
    Then we show that $\dUQMA=\UQMA\cap \mathsf{co}\text{-}\UQMA$, which implies that $\dUQMA$ is an alternative formulation of $\UQMA \cap \mathsf{co}\text{-}\UQMA$.
    The original definition of $\UQMA\cap \mathsf{co}\text{-}\UQMA$ has two verification circuits, so it is not straightforward to embed the verification circuit into a Hamiltonian using the standard circuit-to-Hamiltonian construction.
    This alternative formulation $\dUQMA$ allows us to solve this difficulty, thereby proving the hardness result.
    }

    {
    First, we provide a definition of $\dUQMA$ (doubly Unique QMA).
    }
    {
    \begin{definition}[\dUQMA]\label{dfn:dUQMA}
        A promise problem $L = (L_{\mathrm{yes}}, L_{\mathrm{no}})$ is in \dUQMA$[c_\text{yes}, c_\text{no}, s]$ if there exists a uniformly generated polynomial-size quantum circuit $U_x$, acting on $l(|x|)$ input qubits and $m(|x|)$ ancilla qubits initialized to $\ket{0^m}$, such that:
        \begin{itemize}
            \item If $x \in L_{\text{yes}}$, there exists $\ket{w_{\text{yes}}}$ s.t. $\left\| (\ket{1}\bra{1}_{\text{out}1}\otimes \ket{1}\bra{1}_{\text{out}2}\otimes I) U_x \ket{w_{\text{yes}}} \otimes \ket{0^m} \right\|^2 \geq c_\text{yes}$  and for any $\ket{\phi} \perp \ket{w_{\text{yes}}}$, $\left\| (I\otimes \ket{1}\bra{1}_{\text{out2}}\otimes I) U_x \ket{\phi} \otimes \ket{0^m} \right\|^2 \leq s$,
            \item If $x \in L_{\text{no}}$, there exists $\ket{w_{\text{no}}}$ s.t. $\left\| (\ket{0}\bra{0}_{\text{out}1}\otimes \ket{1}\bra{1}_{\text{out}2}\otimes I) U_x \ket{w_{\text{no}}} \otimes \ket{0^m} \right\|^2 \geq c_\text{no}$ and for any $\ket{\phi} \perp \ket{w_{\text{no}}}$, $\left\| (I\otimes \ket{1}\bra{1}_{\text{out2}}\otimes I) U_x \ket{\phi} \otimes \ket{0^m} \right\|^2 \leq s$,
        \end{itemize}
        We define $\dUQMA \coloneqq \dUQMA[2/3, 2/3, 1/3]$.
    \end{definition}
    }

    {
    Note that $\dUQMA$ has only a single verification circuit while the original definition of $\UQMA \cap \mathsf{co}\text{-}\UQMA$ in Definition~\ref{def:UQMA-co-UQMA} has two verification circuits.
    }

    {Next we show that the failure probability of \dUQMA\ can be suppressed exponentially by increasing the witness register by only one flag qubit.}

    \begin{lemma}
        \label{thm:dUQMA_amplification_restatement}
        Let $n$ be an input size.
        Then
        \begin{equation*}
            \dUQMA[2/3, 2/3, 1/3] = \dUQMA[1-\mathcal{O}(2^{-n}), 1-\mathcal{O}(2^{-n}), \mathcal{O}(2^{-n})].
        \end{equation*}
    \end{lemma}
    {
    \noindent\textbf{Proof overview.}
    The proof proceeds in three steps.
    First, we view a \dUQMA\ verifier as two ordinary unique-witness verifiers: accepting on output $(1,1)$ gives a \UQMA\ verifier, while accepting on output $(0,1)$ gives a $\mathsf{co}\text{-}\UQMA$ verifier.
    The only nontrivial point in this step is to prove a constant soundness gap for these derived verifiers.
    Second, we apply witness-preserving Marriott--Watrous amplification to these two derived verifiers, reducing both completeness and soundness errors to $\mathcal{O}(2^{-n})$ while preserving the unique accepting witnesses.
    Finally, we add one flag qubit to the witness register and use it to choose between the amplified \UQMA\ and $\mathsf{co}\text{-}\UQMA$ verifiers, thereby constructing a single amplified \dUQMA\ verifier.
    }

    \begin{proof}

 First,
        $\dUQMA[1-\mathcal{O}(2^{-n}), 1-\mathcal{O}(2^{-n}), \mathcal{O}(2^{-n})] \subseteq \dUQMA[2/3, 2/3, 1/3]$ is trivial.
        Thus we prove $\dUQMA[2/3, 2/3, 1/3] \subseteq \dUQMA[1-\mathcal{O}(2^{-n}), 1-\mathcal{O}(2^{-n}), \mathcal{O}(2^{-n})]$ in the following.

        {
         \proofpara{Extracting ordinary unique-witness verifiers}
            Let $L=(L_\text{yes}, L_\text{no})$ be a promise problem such that
        \[
            L \in \dUQMA[2/3, 2/3, 1/3].
        \]
            We observe that its verifier can be used to construct $\UQMA$ and $\mathsf{co}\text{-}\UQMA$ verifiers with $2/3$ completeness and $5/8$ soundness.
        In fact, if we decide to accept when the output qubits are $(\text{out}_1,\text{out}_2)=(1,1)$ (resp. $(0,1)$), then $\dUQMA[2/3, 2/3, 1/3]$ can be seen as $\UQMA$ (resp. $\mathsf{co}\text{-}\UQMA$) with $2/3$ completeness and $5/8$ soundness.
        The completeness conditions for these derived \UQMA\ and $\mathsf{co}\text{-}\UQMA$ verifiers follow immediately from Definition~\ref{dfn:dUQMA}. We therefore focus on proving soundness for the derived \UQMA\ verifier; the $\mathsf{co}\text{-}\UQMA$ case follows by symmetry.
        Let $x \in L_\text{no}$ and $\Pi_{ab}=\ket{a}\bra{a}_{\text{out}1}\otimes \ket{b}\bra{b}_{\text{out}2}\otimes I$ for $a,b \in \{0,1\}$.
        Then any witness state can be represented as $\ket{\Psi}=\alpha\ket{w_\text{no}}+\beta\ket{\phi}$, where $|\alpha|^2+|\beta|^2=1$, $\|\Pi_{01}U_x\ket{w_\text{no}}\ket{0^m}\|^2\geq 2/3$ and $\|(\Pi_{11}+\Pi_{01})U_x\ket{\phi}\ket{0^m}\|^2\leq 1/3$.
        To show the soundness, we consider two cases:
        \proofpara{Case 1}
        Assume that $\|\Pi_{11}U_x\ket{w_\text{no}}\ket{0^m}\|^2+\|\Pi_{11}U_x\ket{\phi}\ket{0^m}\|^2\leq 2/3-c$, where $c$ is sufficiently small fixed constant which we specify later.
        Then
        \begin{align*}
            \|\Pi_{11}U_x\ket{\Psi}\ket{0^m}\|^2
            &=\| \alpha \Pi_{11}U_x\ket{w_\text{no}}\ket{0^m}+\beta \Pi_{11}U_x\ket{\phi}\ket{0^m}\|^2 \\
            &\leq \left( |\alpha|\|\Pi_{11}U_x\ket{w_\text{no}}\ket{0^m}\| + |\beta| \|\Pi_{11}U_x\ket{\phi}\ket{0^m}\|  \right)^2 \\
            &\leq \|\Pi_{11}U_x\ket{w_\text{no}}\ket{0^m}\|^2 + \|\Pi_{11}U_x\ket{\phi}\ket{0^m}\|^2 \\
            &\leq \frac{2}{3} - c.
        \end{align*}
        Here the second inequality follows from Cauchy–Schwarz inequality.
        \proofpara{Case 2}
        Assume that $\|\Pi_{11}U_x\ket{w_\text{no}}\ket{0^m}\|^2+\|\Pi_{11}U_x\ket{\phi}\ket{0^m}\|^2 > 2/3-c$.
        Then
        \begin{align*}
            &\|\Pi_{11}U_x\ket{\Psi}\ket{0^m}\|^2 \\
            &=\| \alpha \Pi_{11}U_x\ket{w_\text{no}}\ket{0^m}+\beta \Pi_{11}U_x\ket{\phi}\ket{0^m}\|^2 \\
            &\leq |\alpha|^2 \|\Pi_{11}U_x\ket{w_\text{no}}\ket{0^m}\|^2 + |\beta|^2 \|\Pi_{11}U_x\ket{\phi}\ket{0^m}\|^2 + 2|\alpha \beta |
            \left| \bra{w_\text{no}}\bra{0^m} U_x^\dagger \Pi_{11} U_x\ket{\phi}\ket{0^m}\right| \\
            &\leq \frac{1}{3} + \left| \bra{w_\text{no}}\bra{0^m} U_x^\dagger \Pi_{11} U_x\ket{\phi}\ket{0^m}\right|.
        \end{align*}
        Since $\braket{w_\text{no}|\phi}=0$, we have
        \begin{align*}
            & \left| \bra{w_\text{no}}\bra{0^m} U_x^\dagger \Pi_{11} U_x\ket{\phi}\ket{0^m}\right| \\
            \leq & \left| \bra{w_\text{no}}\bra{0^m} U_x^\dagger \Pi_{01} U_x\ket{\phi}\ket{0^m}\right| \\
            &\quad + \left| \bra{w_\text{no}}\bra{0^m} U_x^\dagger (\Pi_{10} + \Pi_{00}) U_x\ket{\phi}\ket{0^m}\right| \\
            \leq & \|\Pi_{01} U_x \ket{w_\text{no}}\ket{0^m} \| \cdot \|\Pi_{01} U_x\ket{\phi}\ket{0^m}\| \\
            &\quad + \|(\Pi_{10} + \Pi_{00}) U_x \ket{w_\text{no}}\ket{0^m} \|
                \cdot \|(\Pi_{10} + \Pi_{00}) U_x\ket{\phi}\ket{0^m}\| \\
            \leq & \|\Pi_{01} U_x\ket{\phi}\ket{0^m}\| + \|(\Pi_{10} + \Pi_{00}) U_x \ket{w_\text{no}}\ket{0^m} \| \\
            \leq & \sqrt{\frac{1}{3} - \|\Pi_{11} U_x \ket{\phi} \ket{0^m} \|^2} \\
            &\quad + \sqrt{1 - \| \Pi_{11} U_x \ket{w_\text{no}}\ket{0^m} \|^2
                - \| \Pi_{01} U_x \ket{w_\text{no}}\ket{0^m} \|^2} \\
            \leq & \sqrt{\frac{1}{3} - \|\Pi_{11} U_x \ket{\phi} \ket{0^m} \|^2} \\
            &\quad + \sqrt{\frac{1}{3} - \| \Pi_{11} U_x \ket{w_\text{no}}\ket{0^m} \|^2} \\
            \leq & \sqrt{2\left( \frac{2}{3} - \|\Pi_{11} U_x \ket{\phi} \ket{0^m} \|^2 - \| \Pi_{11} U_x \ket{w_\text{no}}\ket{0^m} \|^2 \right)} \\
            \leq & \sqrt{2c}.
        \end{align*}
        Here we have used Cauchy–Schwarz inequality and given inequalities $\|\Pi_{01}U_x\ket{w_\text{no}}\ket{0^m}\|^2\geq 2/3$ and $\|(\Pi_{11}+\Pi_{01})U_x\ket{\phi}\ket{0^m}\|^2\leq 1/3$.
        We can conclude that
        \begin{equation*}
            \|\Pi_{11}U_x\ket{\Psi}\ket{0^m}\|^2 \leq \frac{1}{3} + \sqrt{2c}.
        \end{equation*}
        }

        \proofpara{{Choosing a constant soundness gap}}
        {
        Putting it altogether, in the NO case,
        \begin{equation*}
            \|\Pi_{11}U_x\ket{\Psi}\ket{0^m}\|^2 \leq \max{\left( \frac{2}{3} - c, \frac{1}{3} + \sqrt{2c} \right)},
        \end{equation*}
        for sufficiently small fixed constant $c>0$.
        If we choose $c=1/24$, then
        \begin{equation*}
            \|\Pi_{11}U_x\ket{\Psi}\ket{0^m}\|^2 \leq \max{\left( \frac{5}{8}, \frac{1}{3} + \frac{1}{2\sqrt{3}} \right)} = \frac{5}{8} < \frac{2}{3},
        \end{equation*}
        which implies $5/8$ soundness.
        }

        \proofpara{{Amplifying and recombining the two verifiers}}
        {
        Next, we will construct the amplified $\dUQMA$ verifier.
        From the above discussion, we have a pair of verification circuits for $\UQMA$ and $\mathsf{co}\text{-}\UQMA$ with $2/3$ completeness and $5/8$ soundness.
        {By using the witness-preserving \textsf{QMA} amplification scheme of Marriott and Watrous~\cite{marriott2005quantum}, we may assume that the completeness error and the soundness error of the verification circuits are at most $\mathcal{O}(2^{-n})$. This amplification acts on the same witness register and transforms the acceptance operator by a polynomial function, so the unique accepting witness and the exponentially small acceptance bound on all states orthogonal to it are preserved.}
        Let $U_\text{yes}$ and $U_\text{no}$ be amplified verification circuits for $\UQMA$ and $\mathsf{co}\text{-}\UQMA$.
        Let $\mathcal{H}$ be the Hilbert space to which the original witness state belongs.
        Then we let the witness space be $\mathbb{C}^2 \otimes \mathcal{H}$ by adding one ancilla qubit.
        We construct a $\dUQMA[1-\mathcal{O}(2^{-n}), 1-\mathcal{O}(2^{-n}), \mathcal{O}(2^{-n})]$ verification circuit $U$ as follows:
        \begin{enumerate}
            \item Measure the first ancilla qubit in the computational bases, obtaining $a\in \{0,1\}$.
            \item If $a=1$, perform $U_{\text{yes}}$ on the remaining witness register and measure its output qubit, obtaining $b\in \{0,1\}$. Output $(\text{out}_1,\text{out}_2)=(1,b)$.
            \item If $a=0$, perform $U_{\text{no}}$ on the remaining witness register and measure its output qubit, obtaining $b'\in \{0,1\}$. Output $(\text{out}_1,\text{out}_2)=(0,b')$.
        \end{enumerate}
            As shown below, $U$ satisfies the conditions of a
        $\dUQMA[1-\mathcal{O}(2^{-n}), 1-\mathcal{O}(2^{-n}), \mathcal{O}(2^{-n})]$ verifier. Thus,
        \[
            L\in \dUQMA[1-\mathcal{O}(2^{-n}), 1-\mathcal{O}(2^{-n}), \mathcal{O}(2^{-n})],
        \]
        which implies
        \[
            \dUQMA[2/3, 2/3, 1/3]
            \subseteq
            \dUQMA[1-\mathcal{O}(2^{-n}), 1-\mathcal{O}(2^{-n}), \mathcal{O}(2^{-n})].
        \]
        \proofpara{{Yes instances}}
        Assume $x\in L_{\mathrm{yes}}$.
        Let $\ket{w_{\mathrm{yes}}}$ be the unique $\UQMA$ witness for $U_{\mathrm{yes}}$.
        Then, for the witness $\ket{1}\ket{w_{\mathrm{yes}}}$, $U$ outputs $(1,1)$ with a probability at least $1-\mathcal{O}(2^{-n})$.
        }

        {
        Next we show uniqueness.
        Let $\ket{\Phi}\in \mathbb{C}^2\otimes \mathcal{H}$ be any quantum state orthogonal to  $\ket{1}\ket{w_{\mathrm{yes}}}$.
        We can write $\ket{\Phi} = \ket{1}\ket{\phi_1}+\ket{0}\ket{\phi_0}$, where $\ket{\phi_1}, \ket{\phi_0}\in \mathcal{H}$ are unnormalized states and satisfy $\|\ket{\phi_1}\|^2+\|\ket{\phi_0}\|^2=1$ and $\braket{w_{\mathrm{yes}}|\phi_1}=0$.
        {Then we can see that the probability that $U$ outputs $(1,1)$ is upper-bounded by $\| \Pi_1 U_{\text{yes}} \ket{\phi_1}\ket{0^m} \|^2 \leq \mathcal{O}(2^{-n})\|\ket{\phi_1}\|^2 \leq \mathcal{O}(2^{-n})$ by the construction of the verifier.}
        {Also the probability that $U$ outputs $(0,1)$ is upper-bounded by $\| \Pi_1 U_{\text{no}} \ket{\phi_0}\ket{0^m} \|^2 \leq \mathcal{O}(2^{-n})\|\ket{\phi_0}\|^2 \leq \mathcal{O}(2^{-n})$.}
        Therefore, for any state orthogonal to  $\ket{1}\ket{w_{\mathrm{yes}}}$, the second output qubit of the verifier $U$ is $1$ with a probability at most $\mathcal{O}(2^{-n})$.
        }

        \proofpara{{No instances}}
        {
        Assume $x\in L_{\mathrm{no}}$.
        Let $\ket{w_{\mathrm{no}}}$ be the unique $\mathsf{co}\text{-}\UQMA$ witness for $U_{\mathrm{no}}$.
        Then, for the witness $\ket{0}\ket{w_{\mathrm{no}}}$, $U$ outputs $(0,1)$ with a probability at least $1-\mathcal{O}(2^{-n})$.
        }

        {
        Next we show uniqueness.
        Let $\ket{\Phi}\in \mathbb{C}^2\otimes \mathcal{H}$ be any quantum state orthogonal to  $\ket{0}\ket{w_{\mathrm{no}}}$.
        We can write $\ket{\Phi} = \ket{1}\ket{\phi_1}+\ket{0}\ket{\phi_0}$, where $\ket{\phi_1}, \ket{\phi_0}\in \mathcal{H}$ are unnormalized states and satisfy $\|\ket{\phi_1}\|^2+\|\ket{\phi_0}\|^2=1$ and $\braket{w_{\mathrm{no}}|\phi_0}=0$.
        {Then we can see that the probability that $U$ outputs $(1,1)$ is upper-bounded by $\| \Pi_1 U_{\text{yes}} \ket{\phi_1}\ket{0^m} \|^2 \leq \mathcal{O}(2^{-n})\|\ket{\phi_1}\|^2 \leq \mathcal{O}(2^{-n})$ by the construction of the verifier.}
        {Also the probability that $U$ outputs $(0,1)$ is upper-bounded by $\| \Pi_1 U_{\text{no}} \ket{\phi_0}\ket{0^m} \|^2 \leq \mathcal{O}(2^{-n})\|\ket{\phi_0}\|^2 \leq \mathcal{O}(2^{-n})$.}
        Therefore, for any state orthogonal to  $\ket{0}\ket{w_{\mathrm{no}}}$, the second output qubit of the verifier $U$ is $1$ with a probability at most $\mathcal{O}(2^{-n})$.
        }

    \end{proof}

    {
    Finally, we show that $\dUQMA$ is equivalent to $\UQMA \cap \mathsf{co}\text{-}\UQMA$.
    }
    {
    \begin{lemma}
        \begin{equation*}
            \dUQMA = \UQMA \cap \mathsf{co}\text{-}\UQMA .
        \end{equation*}
    \end{lemma}
    }
    \begin{proof}
        {
        If we decide to accept when the output qubits are $\ket{1}_{\text{out}1}\ket{1}_{\text{out}2}$, then \dUQMA\ can be seen as a special case of \UQMA~considering the amplified verifier of \dUQMA~as shown in Lemma~\ref{thm:dUQMA_amplification_restatement}.
        Similarly, if we accept when the output qubits are $\ket{0}_{\text{out}1}\ket{1}_{\text{out}2}$, then \dUQMA\ is a special case of $\mathsf{co}\text{-}\UQMA$.
        Therefore, $\dUQMA \subseteq \UQMA \cap \mathsf{co}\text{-}\UQMA$ holds.
        }

        {
        Now we confirm the opposite direction $\UQMA \cap \mathsf{co}\text{-}\UQMA \subseteq \dUQMA$.
        Let a promise problem $L=(L_\text{yes}, L_\text{no}) \in \UQMA \cap \mathsf{co}\text{-}\UQMA$.
        Then we have a pair of  verification circuits $U_{\text{yes}}$ and $U_{\text{no}}$.
        The circuit $U_{\text{yes}}$ ($U_{\text{no}}$) accepts a unique witness state in the case of a YES (NO) instance and rejects any witness state in the case of a NO (YES) instance.
        In the latter half of the proof of Lemma~\ref{thm:dUQMA_amplification_restatement}, we have shown how to construct a verification circuit for $\dUQMA[1-\mathcal{O}(2^{-n}), 1-\mathcal{O}(2^{-n}), \mathcal{O}(2^{-n})]$ given a pair of verification circuits for $\UQMA$ and $\mathsf{co}\text{-}\UQMA$ with $2/3$ completeness and $5/8$ soundness.
        Using this construction, we can construct a $\dUQMA$ verification circuit from $U_{\text{yes}}$ and $U_{\text{no}}$.
        Thus $L\in \dUQMA$, which implies $\UQMA \cap \mathrm{co}\text{-}\UQMA \subseteq \dUQMA$.
        }

    \end{proof}

    \section{Proof of Theorem~\ref{thm:BQPcomp}: Showing BQP-Completeness\label{sec:BQPcomp}}

    \subsection{Containment in \textsf{BQP}\label{sec:BQP_containment}}
    \begin{proof}
        This result follows from Theorem~\ref{thm:BPE_algorithm} and the algorithm presented in Section~\ref{sec:improved_algorithm} together with known ground-state preparation methods~\cite{ge2019faster,lin2020near}.
        Given a polynomial-size classical description of a guiding state \( \ket{c} \) that has inverse-polynomial overlap with the ground state \( \ket{\psi(0)} \), one can prepare \( \ket{\psi(0)} \) in polynomial time, even without prior knowledge of the ground-state energy, by exploiting the spectral gap condition of \(H(0)\) (see, e.g., Corollary~9 in Ref.~\cite{lin2020near}). Once the ground state is prepared, we apply quantum phase estimation to two rescaled adiabatic evolutions, as described in Section~\ref{sec:improved_algorithm}.
        This procedure yields an estimate \( \widehat{\theta}_B \) of the Berry phase that satisfies
        $|\widehat{\theta}_B - \theta_B| \leq \varepsilon_B$
        with high probability greater than \(2/3\), where \( \varepsilon_B \ge 1/\mathrm{poly}(n) \).
        The entire procedure uses polynomial circuit depth and gate complexity under the assumptions of Definition~\ref{def:BPE_guided}.

        To solve the decision problem in Definition~\ref{def:BPE_guided}, we set the target precision so that \( \varepsilon_B < {\delta} \).
        Under the promise that \( \theta_B \) lies either in \( [a+\delta,b-\delta]_{2\pi} \) or in \( [b+\delta,a-\delta]_{2\pi} \), this accuracy ensures that the algorithm can determine which interval contains \( \theta_B \).
        Concretely, if \( \widehat{\theta}_B \) lies within distance \( \delta-\varepsilon_B \) of \( [a,b]_{2\pi} \), the algorithm outputs \(1\); otherwise, it outputs \(0\).

        Therefore, the estimation algorithm directly solves the promised decision problem with constant success probability greater than \(2/3\), establishing containment in \(\mathsf{BQP}\).
    \end{proof}

    \subsection{\BQP-hardness}\label{sec:BQP_hardness}
    \begin{proof}
\proofpara{{Construction of the parameterized Hamiltonian}}
        Let $x$ be an instance of a \BQP\ problem.
        {Let \(U_x=U_TU_{T-1}\cdots U_1\underbrace{I\cdots I}_{M}\) be a
        quantum circuit with pre-amplified success probability and \(M\)-step initial idling
        on \(n\) qubits such that}
        \begin{itemize}
            \item If $x \in L_{\mathrm{yes}}$, then $p_1= \| \Pi_1 U_x \ket{0^n} \|^2 = 1 - \mathcal{O}(2^{-n})$,
            \item If $x \in L_{\mathrm{no}}$, then $p_1  =\mathcal{O}(2^{-n})$.
        \end{itemize}
        The $n$-qubit register is initialized to $\ket{0^n}$.
        {Here, $M$ is chosen s.t. $T/(T+M+1)\leq 1/\mathrm{poly}(n)$.}

        Let us consider a standard 5-local circuit-to-Hamiltonian construction \cite{kitaev2002classical} (without the output term) $H_{\mathrm{hist}}=H_\text{in}+H_\text{prop}+H_\text{clock}$ on $\mathcal{H}_n\otimes \mathcal{H}_\text{clock}$ where
        \begin{align*}
             & H_{\text{in}} \coloneqq \sum_i \ket{1}\bra{1}_{A_i}  \otimes \ket{0}\bra{0}_{C_1}                  \\
             & H_{\text{clock}} \coloneqq \sum_{t=1}^{T+M-1} \ket{0}\bra{0}_{C_t}\otimes \ket{1}\bra{1}_{C_{t+1}} \\
             & H_{\text{prop}} \coloneqq
            \sum_{t=1}^{T+M} H_t
        \end{align*}
        where
        $$H_t = -\frac{1}{2}U_t\otimes\ket{t}\bra{t-1}_C-\frac{1}{2} U_{t}^\dagger \otimes \ket{t-1}\bra{t}_C +\frac{1}{2}(\ket{t}\bra{t}_C+\ket{t-1}\bra{t-1}_C). $$
        Here, $\mathcal{H}_n$ is the $n$-qubit Hilbert space and $\mathcal{H}_\text{clock}$ is the clock space (time is represented by unary), $C_i$ represents the $i$-th qubit in the clock register.

        The non-degenerate ground state of $H_{\mathrm{hist}}$ is given by
        $$
            \ket{\psi_{\text{hist}}}\coloneqq
            \frac{1}{\sqrt{T+M+1}}
            \sum_{t=0}^{T+M}
            U_tU_{t-1}...U_1 \ket{0^n}\otimes \ket{t},
        $$
        where
        {$U_t=I_{n+m}$ for $1\leq t\leq M$, and \(U_{M+j}\) is the \(j\)-th gate of the original amplified computation for \(1\leq j\leq T\).}
        {Let \(N\coloneqq T+M+1\).}
        Note that $H_{\mathrm{hist}}$ has a spectral gap $\Delta=\Omega\left(\frac{1}{(T+M)^3}\right)$ \cite{gharibian2012hardness}.

        Let $H_r(\lambda)$  be
        $$
            H_r(\lambda) \coloneqq H_{\mathrm{hist}} + r V(\lambda)
        $$
        {where \(r\) is an inverse-polynomial parameter chosen below}, and
        \[
            {
            V(\lambda):=
            \left(
            e^{2\pi i \lambda}\ket{1}\bra{0}_\text{out}
            + e^{-2\pi i \lambda} \ket{0}\bra{1}_{\text{out}}
            \right)
            \otimes \ket{1}\bra{1}_{C_{T+M}}.
            }
        \]
        {Here, the subscript ``out'' indicates that the operator acts on the output qubit, and \(C_{T+M}\) is the final clock qubit. On the legal clock subspace, \(\ket{1}\bra{1}_{C_{T+M}}\) coincides with the projector onto the final clock state \(\ket{T+M}\). We also use the same subscript notation for states (e.g., $\ket{0}_{\mathrm{out}}$ denotes the state of the output qubit).}
        {
                It can be seen that
                    {$H_r(0)=H_r(1)$ and}
                $$\left\|\frac{\mathrm{d}H_r(\lambda)}{\mathrm{d}\lambda}\right\|\in \mathcal{O}(1) \text{ \ and \ } \left\|\frac{\mathrm{d}^2H_r(\lambda)}{\mathrm{d}\lambda^2}\right\|\in \mathcal{O}(1)$$
                for all $\lambda\in  \mathbb{R}$.
            }

\proofpara{{Analysis of parameterized ground states}}
        {By taking \(M\in \mathrm{poly}(n)\) sufficiently large,
        \(\ket{\psi_{\text{hist}}}\)
        has large overlap with the guiding state}
        \begin{align}
        \label{eq:guiding}
            {\ket{c} = \ket{0^n}\otimes \ket{[0,M]}}
        \end{align}
        {where}
        $$
            {\ket{[0,M]}:= \frac{1}{\sqrt{M+1}}\sum_{t=0}^{M}  \ket{t}.}
        $$
        {Let \(\eta_M\coloneqq \sqrt{T/N}\). By choosing \(M\in\mathrm{poly}(n)\) sufficiently large, \(\eta_M\) can be made inverse-polynomially small. Then}
        \begin{align*}
            {\ket{\psi_\text{hist}} =
                \ket{0^n}\otimes \ket{[0,M]}+\mathcal{O}(\eta_M).}
        \end{align*}

        Recall that
        the eigenvalue of $\ket{\psi_{\text{hist}}}$ is $0$.
        With the second-order perturbation theory, the normalized ground state of $H_r(\lambda)$ can be written as
        \begin{align*}
            \ket{\psi_r(\lambda)}= &
            \left(
            1-\frac{r^2}{2}\sum_{k\neq0}\frac{|\bra{\psi^{(k)}}V\ket{\psi_{\text{hist}}}|^2}{(E_k)^2}
            \right)
            \ket{\psi_{\text{hist}}}- r\sum_{k\neq 0} \frac{\bra{\psi^{(k)}}V(\lambda)\ket{\psi_\text{hist}}}{ E_k} \ket{\psi^{(k)}} \\&
            + r^2\ket{\perp(\lambda)} + \mathcal{O}(r^3)
        \end{align*}
        where $E_k$ is the $k$-th eigenvalue of {$H_{\mathrm{hist}}$} and $\ket{\psi^{(k)}}$ is the corresponding eigenvector ($\ket{\psi^{(0)}}= \ket{\psi_\text{hist}}$), and
            \begin{align*}
                \ket{\perp(\lambda)}
                ={}& \sum_{k\neq 0,k'\neq 0} \ket{\psi^{(k)}}
                \frac{\bra{\psi^{(k)}}V\ket{\psi^{(k')}}}{E_k}
                \frac{\bra{\psi^{(k')}}V\ket{\psi_\text{hist}}}{E_{k'}} \\
                &- \sum_{k\neq 0}\ket{\psi^{(k)}}
                \frac{\bra{\psi^{(k)}}V\ket{\psi_\text{hist}}}{(E_k)^2}
                \bra{\psi_\text{hist}}V\ket{\psi_\text{hist}}.
            \end{align*}

\proofpara{{Analysis of Berry curvature and Berry phase}}
        Now, it can be seen that
        \begin{align*}
            \frac{\mathrm{d}}{\mathrm{d}\lambda} \ket{\psi_r(\lambda)}= &
            {-}r^2\sum_{k\neq0}\frac{ \mathrm{Re}\Big(
                \bra{\psi_{\text{hist}}}V\ket{\psi^{(k)}}
                \bra{\psi^{(k)}}\frac{\mathrm{d} V(\lambda)}{\mathrm{d} \lambda}\ket{\psi_{\text{hist}}}\Big)
            }{(E_k)^2}
            \ket{\psi_{\text{hist}}}                                      \\&
            -
            r\sum_{k\neq 0} \frac{\bra{\psi^{(k)}}\frac{\mathrm{d} V(\lambda)}{\mathrm{d} \lambda}\ket{\psi_{\text{hist}}}}{ E_k} \ket{\psi^{(k)}} + r^2  \frac{\mathrm{d}}{\mathrm{d}\lambda} \ket{\perp(\lambda)}+ \mathcal{O}(r^3).
        \end{align*}
        It can be easily seen that
        $$
            \bra{\psi_{\text{hist}}}
            \frac{\mathrm{d}}{\mathrm{d}\lambda} \ket{\perp(\lambda)}=0 .
        $$
        Therefore,
        \begin{align*}
            \bra{\psi_r(\lambda)}\frac{\mathrm{d}}{\mathrm{d}\lambda} \ket{\psi_r(\lambda)}= &
            {-}r^2\sum_{k\neq0}\frac{ \mathrm{Re}\Big(
                \bra{\psi_{\text{hist}}}V\ket{\psi^{(k)}}
                \bra{\psi^{(k)}}\frac{\mathrm{d} V(\lambda)}{\mathrm{d} \lambda}\ket{\psi_{\text{hist}}}\Big)
            }{(E_k)^2}                                                                                                        \\
                                                                                                 & +  r^2 \sum_{k\neq 0,k'\neq 0}
                \frac{{\bra{\psi_\text{hist}}V(\lambda)\ket{\psi^{(k)}}}}{E_k} \\
                                                                                                 &\qquad \times
                \frac{\bra{\psi^{(k')}}\frac{\mathrm{d} V(\lambda)}{\mathrm{d} \lambda}\ket{\psi_{\text{hist}}}}{E_{k'}}
                \braket{\psi^{(k)}|\psi^{(k')}}+\mathcal{O}(r^3)
                \\
                                                                                             & =
            {i} r^2 \mathrm{Im}\left(
            \bra{\psi_{\text{hist}}}
            V(\lambda)
            \left(
            \sum_{{k\neq 0}}
            \frac{\ket{\psi^{(k)}}\bra{\psi^{(k)}}}{(E_k)^2}
            \right)
            \frac{\mathrm{d} V(\lambda)}{
                \mathrm{d} \lambda}\ket{\psi_{\text{hist}}}\right)
            +\mathcal{O}(r^3).
        \end{align*}

        {Here and below, \(\mathcal{O}(2^{-n})\) denotes a vector of exponentially small norm.
        }
        It can be seen that
        \begin{align*}
            V(\lambda)\ket{\psi_{\text{hist}}}=
            \begin{cases}
                {\frac{e^{-2\pi i \lambda}}{\sqrt N}  \ket{0}_{\text{out}}\otimes \ket{g_0} \otimes \ket{T+M}+ \mathcal{O}(2^{-n}),}    & {\text{if } x\in L_{\text{yes}}} \\
                {\frac{e^{2\pi i \lambda}}{\sqrt N}\ \   \ket{1}_{\text{out}}\otimes \ket{g_1} \otimes \ket{T+M}+ \mathcal{O}(2^{-n}),} & {\text{if } x\in L_{\text{no}}.}
            \end{cases}
        \end{align*}
        and
        \begin{align*}
            \frac{\mathrm{d} V(\lambda)}{
                \mathrm{d} \lambda}\ket{\psi_{\text{hist}}}=
            \begin{cases}
                {-\frac{2\pi ie^{-2\pi i \lambda}}{\sqrt N}  \ket{0}_{\text{out}}\otimes \ket{g_0} \otimes \ket{T+M}+ \mathcal{O}(2^{-n}),}
                & {\text{if } x\in L_{\text{yes}}} \\
                {\ \  \frac{2\pi i  e^{2\pi i \lambda}}{\sqrt N} \ \  \ket{1}_{\text{out}}\otimes \ket{g_1} \otimes \ket{T+M}+ \mathcal{O}(2^{-n}),}
                & {\text{if } x\in L_{\text{no}},}
            \end{cases}
        \end{align*}
        where we have used
        $$
            {\frac{\mathrm{d} V(\lambda)}{\mathrm{d} \lambda}= 2\pi i
            \left(e^{2\pi i \lambda}\ket{1}\bra{0}_\text{out}
            - e^{-2\pi i \lambda}\ket{0}\bra{1}_{\text{out}}\right)
            \otimes \ket{1}\bra{1}_{C_{T+M}}.}
        $$
        Therefore, omitting exponentially small terms, we obtain
        \begin{align*}
             & i\bra{\psi_r(\lambda)}\frac{\mathrm{d}}{\mathrm{d}\lambda} \ket{\psi_r(\lambda)}= \\
             &
            \begin{cases}
                {\ \    \frac{r^2\cdot 2\pi}{N}  \bra{0}_{\mathrm{out}}\bra{g_0}\bra{T+M}\left(
                \sum_{k\neq 0}
                \frac{\ket{\psi^{(k)}}\bra{\psi^{(k)}}}{(E_k)^2}
                \right)\ket{0}_{\mathrm{out}} \ket{g_0}\ket{T+M}
                +\mathcal{O}(r^3),}
                \\
                {-\frac{r^2\cdot 2\pi}{N} \bra{1}_{\mathrm{out}}\bra{g_1}\bra{T+M}\left(
                \sum_{k\neq 0}
                \frac{\ket{\psi^{(k)}}\bra{\psi^{(k)}}}{(E_k)^2}
                \right)\ket{1}_{\mathrm{out}} \ket{g_1}\ket{T+M}
                +\mathcal{O}(r^3)}
            \end{cases}
        \end{align*}
        for $x\in L_{\text{yes}}$ and $x\in L_{\text{no}}$, respectively.
        Observe that
        $
            \Delta\leq E_k \leq K\coloneqq \|H_{\mathrm{hist}}\|={\mathcal{O}(N+n)}
        $ and the following holds:
        $$
            0 \preceq
            \sum_{k\neq 0}
            \frac{\ket{\psi^{(k)}}\bra{\psi^{(k)}}}{K^2}
            \preceq
            \sum_{k\neq 0}
            \frac{\ket{\psi^{(k)}}\bra{\psi^{(k)}}}{(E_k)^2}
            \preceq
            \sum_{k\neq 0}
            \frac{\ket{\psi^{(k)}}\bra{\psi^{(k)}}}{\Delta^2}.
        $$
        This implies for any state $\ket{\phi}\in \mathcal{H}_n\otimes\mathcal{H}_{\mathrm{clock}}$,
        $$
            0\leq
            \bra{\phi}\left( \sum_{k\neq 0}
            \frac{\ket{\psi^{(k)}}\bra{\psi^{(k)}}}{K^2}\right) \ket{\phi}\leq
            \bra{\phi}\left( \sum_{k\neq 0}
            \frac{\ket{\psi^{(k)}}\bra{\psi^{(k)}}}{(E_k)^2}\right) \ket{\phi}\leq
            \bra{\phi}\left( \sum_{k\neq 0}
            \frac{\ket{\psi^{(k)}}\bra{\psi^{(k)}}}{\Delta^2}\right) \ket{\phi}.
        $$
        Now, because
        $\sum_{k\neq 0}\ket{\psi^{(k)}}\bra{\psi^{(k)}}=I-\ket{\psi_{\text{hist}}}\bra{\psi_{\text{hist}}}$,
        $$
            0\leq
            \bra{\phi} \frac{I-\ket{\psi_{\text{hist}}}\bra{\psi_{\text{hist}}}}{K^2} \ket{\phi}
            \leq
            \bra{\phi}\left( \sum_{k\neq 0}
            \frac{\ket{\psi^{(k)}}\bra{\psi^{(k)}}}{(E_k)^2}\right) \ket{\phi}\leq
            \bra{\phi} \frac{I-\ket{\psi_{\text{hist}}}\bra{\psi_{\text{hist}}}}{\Delta^2} \ket{\phi}
        $$
        holds for any $\ket{\phi}\in \mathcal{H}_n\otimes\mathcal{H}_{\mathrm{clock}}$.
        {We can utilize this inequality for
        $\ket{0}_{\mathrm{out}} \ket{g_0}\ket{T+M}$ and $\ket{1}_{\mathrm{out}}  \ket{g_1}\ket{T+M}$ for the YES and NO cases, respectively.}
        It holds that
        $$
            {|\bra{\psi_{\mathrm{hist}}}\ket{0}_{\mathrm{out}} \ket{g_0}\ket{T+M}|=\mathcal{O}(2^{-n})}
        $$
        for $x\in L_{\mathrm{yes}}$ and
        $$
            {|\bra{\psi_{\mathrm{hist}}}\ket{1}_{\mathrm{out}} \ket{g_1}\ket{T+M}| = \mathcal{O}(2^{-n})}
        $$
        for $x\in L_{\mathrm{no}}$.
        Therefore, if $x\in L_{\text{yes}}$,
        $$
            { \frac{2\pi r^2}{N K^2}\left(1-\mathcal{O}(2^{-n})\right)
            -\mathcal{O}(r^3)
            \leq
            i\mathcal{A}_\lambda
            \leq
            \frac{2\pi r^2}{N\Delta^2}
            +\mathcal{O}(r^3)}
        $$
        and if $x\in L_{\text{no}}$,
        $$
            { \frac{2\pi r^2}{N K^2}\left(1-\mathcal{O}(2^{-n})\right)
            -\mathcal{O}(r^3)
            \leq -i\mathcal{A}_\lambda \leq
            \frac{2\pi r^2}{N\Delta^2}
            +\mathcal{O}(r^3),}
        $$
        where $\mathcal{A}_\lambda=\bra{\psi(\lambda)}\frac{\mathrm{d}}{\mathrm{d}\lambda}\ket{\psi(\lambda)}$.
        {
            We choose \(r\leq \Delta/4\) sufficiently small so that the \(\mathcal{O}(r^3)\) term is absorbed into the leading \(r^2/(NK^2)\) lower bound. Then the lower bound is still inverse-polynomially positive, while the upper bound is at most \(\pi/2\). Hence, for some \(\delta>1/\mathrm{poly}(n)\),}
        $$
            \delta \leq {i}\mathcal{A}_\lambda\leq \frac{\pi}{2}
        $$
        if $x\in L_{\text{yes}}$ and
        $$
            \delta \leq {-i}\mathcal{A}_\lambda\leq \frac{\pi}{2}
        $$
        if $x\in L_{\text{no}}$.

        The Berry phase is given by
        $$
            \theta_B= i\int_0^1 \mathcal{A}_\lambda\, \mathrm{d}\lambda.
        $$
        Therefore,
        \begin{equation*}
            \theta_B\in
            \begin{cases}
                {\left[\delta, \frac{\pi}{2}\right]} & \text{if } x\in L_\text{yes} \\
                {\left[\frac{3\pi}{2}, 2\pi-\delta\right]}
                                                     & \text{if } x\in L_\text{no}.
            \end{cases}
        \end{equation*}
        {Corresponding to the definition of the problem, by taking $a={0}$ and $b={\pi}$, $\theta_B \in[a+\delta/2,b-\delta/2]_{2\pi}$ if $x\in\mathcal{L}_{\mathrm{yes}}$ and $\theta_B\in[b+\delta/2,a-\delta/2]_{2\pi}$ if $x\in\mathcal{L}_{\mathrm{no}}$.
        }

    \proofpara{{Analysis of guiding state overlap}}

        {For clarity, we recall the parameter dependency structure:
        \begin{center}
        \fbox{\begin{minipage}{0.9\linewidth}
        \[
        \begin{array}{c}
        T\in \mathrm{poly}(n)\quad\text{is fixed by the input circuit}
        \\[0.4em]
        \Downarrow
        \\[0.4em]
        \text{choose }M\in\mathrm{poly}(n)\text{ so that }T/N\leq 1/\mathrm{poly}(n)
        \quad (N=T+M+1)
        \\[-0.1em]
        \text{{this choice controls the initial-idling guiding-state overlap}}
        \\[0.4em]
        \Downarrow
        \\[0.4em]
        \Delta=\Omega(N^{-3}),\qquad K\coloneqq \|H_{\mathrm{hist}}\|=\mathcal{O}(N+n)
        \\[0.4em]
        \Downarrow
        \\[0.4em]
        \text{choose } r \text{ such that}
        \ r\leq \Delta/4,
        \\[-0.1em]
            \text{and the cubic term is absorbed, i.e.,}
            \\[-0.1em]
            2\pi \frac{r^2}{NK^2}\left(1-\mathcal{O}(2^{-n})\right)-\mathcal{O}(r^3)
        >1/\mathrm{poly}(n).
        \end{array}
        \]
        \end{minipage}}
        \end{center}
        In particular, this choice of \(M\) is not tied to the perturbation strength \(r\).}

    {From the definition of the guiding state $\ket{c}$ in Eq.~\eqref{eq:guiding}, it can be observed that}
        \begin{align}
        \label{eq:overlap}
            |\braket{c|\psi_{\mathrm{hist}}}|^2=
            \frac{1}{(T+M+1)(M+1)}
            \left|
            \sum_{t=0}^{M} \braket{t|t} \right|^2 = \frac{M+1}{T+M+1}
            =1-\frac{T}{N}.
        \end{align}
        {
        By taking \(M\in\mathrm{poly}(n)\) sufficiently large so that \(T/N\leq 1/\mathrm{poly}(n)\), we obtain
        \[
            |\braket{c|\psi_{\mathrm{hist}}}|^2 \geq 1-\frac{1}{\mathrm{poly}(n)}.
        \]
        Since \(\ket{\psi_r(0)}\) is an \(\mathcal{O}(r/\Delta)\)-perturbation of \(\ket{\psi_{\mathrm{hist}}}\) and \(r/\Delta\leq 1/\mathrm{poly}(n)\), the same guiding state satisfies
        \[
            |\braket{c|\psi_r(0)}|^2 \geq 1-\frac{T}{N}-\mathcal{O}(r/\Delta) \geq  1-\frac{1}{\mathrm{poly}(n)}.
        \]
        Thus a classical description of \(\ket{c}\) provides a guiding state with overlap \(\gamma=1-1/\mathrm{poly}(n)\) with the initial ground state of \(H_r(0)\). A simple state
        \(\ket{c'}= \ket{0^n}\otimes \ket{0}_{C}\) also works as a guiding state with overlap \(\gamma \geq 1/\mathrm{poly}(n)\).}

\proofpara{{Spectral gap}}
        Recall
        $
            H_r(\lambda) = H_{\mathrm{hist}} + r V(\lambda)
        $
        and the spectral gap of $H_{\mathrm{hist}}$ is larger than $\Delta\in \Omega(1/(T+M)^3)$.
        Because $\|r V(\lambda)\|=r$, by taking $r\leq \Delta/4$,
        the spectral gap of $H_r(\lambda) = H_{\mathrm{hist}} + r V(\lambda)$ is lower bounded by $\Delta/2$ for any $\lambda$.

    \end{proof}

    \section{Proof of Theorem~\ref{thm:UQMAcomp}: Showing {$\UQMA \cap \mathsf{co}\text{-}\UQMA$}-Completeness}
    \label{sec:UQMA}

    \subsection{Containment in {$\UQMA \cap \mathsf{co}\text{-}\UQMA$}}
    \label{sec:UQMA_containment}

    We now show that the decision version of the Berry phase estimation problem—determining whether the Berry phase $\theta_B$ acquired by the unique ground state of a family of Hamiltonians $\{H(\lambda)\}_{\lambda \in [0,1]}$ lies
    in $[a+\delta,b-\delta]_{2\pi}$ or $[b+\delta,a-\delta]_{2\pi}$—is contained in the complexity class \dUQMA~.
    {
    As shown in Section~\ref{sec:dUQMA_characterization}, $\dUQMA=\UQMA\cap \mathsf{co}\text{-}\UQMA$.
    Thus this result also implies containment in $\UQMA\cap \mathsf{co}\text{-}\UQMA$.
    Note that we can also show containment in $\UQMA\cap \mathsf{co}\text{-}\UQMA$ directly.
    However, the proof for $\dUQMA$ is simpler than that for $\UQMA\cap \mathsf{co}\text{-}\UQMA$ since $\dUQMA$ has a single verifier.
    }

    \begin{proof}

        The class {\dUQMA}
        consists of promise problems for which there exists a unique quantum witness in YES and NO cases.
        In the YES case, the verifier (a polynomial-time quantum circuit) outputs 11 with a probability at least $2/3$, while in the NO case the verifier outputs 01 with a probability at least $2/3$.
        Furthermore, for any state orthogonal to the unique witness, the verifier’s second output qubit is 0 with a probability at least $2/3$.

        \paragraph{Verifier Protocol.}
        The verifier proceeds as follows:
        \begin{enumerate}
            \item \textbf{Receive witness.} Merlin sends a quantum state $|\psi\rangle$ claimed to be the ground state of $H(0)$.

            \item \textbf{Ground state verification.} The verifier estimates the energy of $H(0)$ with respect to $|\psi\rangle$ within an additive error {$< \Delta_\mathrm{min}/8$} using quantum phase estimation algorithm with a success probability ${\geq}5/6$.
                  Then, if the estimate is smaller than {$E_\text{th}$}, the verifier proceeds to the next step.
                  Otherwise, the verifier outputs 0 for the second output qubit and ends the verification procedure.
                  Since $H(0)$ is non-degenerate and gapped, this test uniquely identifies the ground state, i.e., the unique witness state that is separated by a constant factor $\geq 1/3$ in acceptance probability of the second output from all states that are orthogonal to itself.
            \item \textbf{Berry phase estimation.}
            The verifier then runs the polynomial-time quantum algorithm for Berry phase estimation described in Section~\ref{sec:improved_algorithm}.
                  Using the improved two-speed protocol (with full $2\pi$ phase access), the verifier performs quantum phase estimation on the unitary evolutions $U(T,0)$ and $U_N(T,0)$ to extract estimates $\phi^{(1)}$ and $\phi^{(N)}$, and recovers $\theta_B$ to within additive error {$\varepsilon_B<\delta$, where $\delta\geq 1/\mathrm{poly}(n)$}.
                  This procedure succeeds with a probability at least $5/6$, which leads to the overall failure probability at most $1/3$.

            \item \textbf{Decision test.} If the estimate $\hat{\theta}_B$ lies in {$[a,b]_{2\pi}$}, the verifier accepts (outputs 11). {Otherwise}, the verifier rejects (outputs 01).
        \end{enumerate}
        By the promise on the spectral gap and the adiabatic condition, the adiabatic evolution remains in the ground state manifold within a constant error.
        Since the ground state is non-degenerate, any incorrect witness is rejected during the energy test with high probability.
        Conversely, the correct witness yields accurate Berry phase estimates, and the decision is made reliably with an inverse-polynomial promise gap.
        Thus, the problem falls into \dUQMA.
    \end{proof}

    \subsection{{$\UQMA \cap \mathsf{co}\text{-}\UQMA$}-hardness}
    \label{sec:UQMA_hardness}

    {
    Now we show $\dUQMA$-hardness, which implies $\UQMA \cap \mathsf{co}\text{-}\UQMA$-hardness.
    As discussed in Section~\ref{sec:dUQMA_characterization}, $\dUQMA$ is more suitable for this proof since $\dUQMA$ has a single verifier.
    }

    \begin{proof}

\proofpara{{Construction of the non-parameterized Hamiltonian}}
        Let $x$ be an instance of a \dUQMA\ problem.
        Suppose that $U_x=U_T U_{T-1}\cdots U_1$ is a \dUQMA\ verification circuit on $w+m$ qubits
        with amplified gap (Lemma~\ref{thm:dUQMA_amplification_restatement}) s.t.
        \begin{itemize}
            \item If $x \in L_{\text{yes}}$, there exists $\ket{w_{\text{yes}}}$ s.t. $\left\| (\ket{1}\bra{1}_{\text{out}1}\otimes \ket{1}\bra{1}_{\text{out}2}\otimes I) U_x \ket{w_{\text{yes}}} \otimes \ket{0^m} \right\|^2 \geq 1-O(2^{-n})$  and for any $\ket{\phi} \perp \ket{w_{\text{yes}}}$, $\left\| (I\otimes \ket{1}\bra{1}_{\text{out2}}\otimes I) U_x \ket{\phi} \otimes \ket{0^m} \right\|^2 \leq O(2^{-n})$,
            \item If $x \in L_{\text{no}}$, there exists $\ket{w_{\text{no}}}$ s.t. $\left\| (\ket{0}\bra{0}_{\text{out}1}\otimes \ket{1}\bra{1}_{\text{out}2}\otimes I) U_x \ket{w_{\text{no}}} \otimes \ket{0^m} \right\|^2 \geq 1-O(2^{-n})$ and for any $\ket{\phi} \perp \ket{w_{\text{no}}}$, $\left\| (I\otimes \ket{1}\bra{1}_{\text{out2}}\otimes I) U_x \ket{\phi} \otimes \ket{0^m} \right\|^2 \leq O(2^{-n})$,
        \end{itemize}
        The $m$-qubit register is initialized to $\ket{0^m}$ and the $w$-qubit register is for the witness.
        {Let \(N\coloneqq T+1\).}

        Let us consider a standard 5-local circuit-to-Hamiltonian construction \cite{kitaev2002classical} $H_0=H_\text{in}+H_\text{prop}+H_\text{stab}$ where
        \begin{align*}
             & H_{\text{in}} \coloneqq \sum_i \ket{1}\bra{1}_{A_i}  \otimes \ket{0}\bra{0}_{C_1}                 \\
             & H_{\text{stab}} \coloneqq \sum_{t=1}^{{T-1}} \ket{0}\bra{0}_{C_t}\otimes \ket{1}\bra{1}_{C_{t+1}} \\
             & H_{\text{prop}} \coloneqq
            \sum_{t=1}^{{T}} H_t
        \end{align*}
        where
        $$H_t = -\frac{1}{2}U_t\otimes\ket{t}\bra{t-1}_C-\frac{1}{2} U_{t}^\dagger \otimes \ket{t-1}\bra{t}_C +\frac{1}{2}(\ket{t}\bra{t}_C+\ket{t-1}\bra{t-1}_C). $$

        The $2^w$-degenerate ground state of $H_0$ is given by
        $$
            \ket{\psi_{\text{hist}}}\coloneqq
            \frac{1}{\sqrt{{N}}}
            \sum_{t=0}^{{T}}
            U_tU_{t-1}...U_1 \ket{0^m}\otimes \ket{\psi}\otimes \ket{t},
        $$
        for an arbitrary state $\ket{\psi}$.
        Note that $H_0$ has a spectral gap $\Delta_0 \geq \frac{\pi^2}{64{T}^3}$~\cite{gharibian2012hardness, gharibian2019complexity}.

        Now we follow the small-penalty clock construction presented in Ref.~\cite{deshpande2022importance}.
        {For brevity, write $\Pi^{(0)}_{\mathrm{out}2}\coloneqq \ket{0}\bra{0}_{\mathrm{out}2}$ and $\Pi^{(1)}_{\mathrm{out}2}\coloneqq \ket{1}\bra{1}_{\mathrm{out}2}$.}
        We add the penalty term
        \[
            H_1 = H_{\text{out}2}\coloneqq \varepsilon {\Pi^{(0)}_{\mathrm{out}2}}\otimes {\ket{1}\bra{1}_{C_{T}}}
        \]
        choosing $\varepsilon \leq \Delta_0/16$.
        {Note that on the legal clock subspace, $\ket{1}\bra{1}_{C_{T}}$ coincides with the projection onto the clock state $\ket{T}$, so that $H_1$ acts only on the qubit $\mathrm{out}2$ and the final clock qubit.}
        Here we consider the spectral properties of the Hamiltonian $H_0+H_1$.
        To this end, we redefine the accept operator $\mathcal{Q} = (\bra{0^m}\otimes I) \mathcal{U}^\dagger ({\Pi^{(1)}_{\mathrm{out}2}}\otimes I) \mathcal{U} (\ket{0^m}\otimes I)$, where $\mathcal{U}={U_TU_{T-1}\cdots U_1}$.
        Let the eigenvalues of $\mathcal{Q}$ be $\lambda_0 \geq \lambda_1 \geq \dots \geq \lambda_{2^w - 1}$.
        By the conditions on \dUQMA-circuit,
        we have $\lambda_0 \geq 1-O(2^{-n})$ and $\lambda_1\leq O(2^{-n})$ in both YES and NO cases.
        According to the first-order Schrieffer-Wolff transformation, the Hamiltonian $H_0+H_1$ has eigenvalues {$\varepsilon (1-\lambda_i) / N \pm O(\varepsilon^2/\Delta_0)$} (see the proof of Lemma 26 in Ref.~\cite{deshpande2022importance} for the detailed analysis), which can be used to analyze the ground state energy $E_0^{0+1}$, the first excited state energy $E_1^{0+1}$ and the spectral gap $\Delta_{0+1}=E_1^{0+1} - E_0^{0+1}$ of $H_0+H_1$.
        In both YES and NO cases, it is easy to see that the ground state energy is {$$E_0^{0+1} \leq \varepsilon (1-\bra{\widetilde w} \mathcal{Q} \ket{\widetilde w})/N \leq \frac{\varepsilon \mathcal{O}(2^{-n})}{N},$$}
        where {$\ket{\widetilde w}=\ket{w_{\mathrm{yes}}}$ in the YES case and $\ket{\widetilde w}=\ket{w_{\mathrm{no}}}$ in the NO case}.
        Also we obtain
        {$$E_1^{0+1} \geq \varepsilon (1-\lambda_1) / N - O(\varepsilon^2/\Delta_0) \geq \frac{\varepsilon}{N} - O(\varepsilon^2 / \Delta_0),$$}
        which leads to the spectral gap {$\Delta_{0+1} \geq \frac{\varepsilon}{N} - O(\varepsilon^2 / \Delta_0)$}.
        Here, we omitted the exponentially small factor.
        {We have $\Delta_{0+1}\geq 1/\mathrm{poly}(n)$ as long as $\varepsilon/\Delta_0 = o(1/N)$, hence we choose $\varepsilon=o(\Delta_0/N)=o(1/T^4)$.}

\proofpara{{Analysis of the non-degenerate ground state}}
        Now we show that the {\it non-degenerate} ground state of $H_0+H_1$ is well approximated by
        \begin{equation*}
        \ket{\psi_{\text{hist,acc}}}\coloneqq
            \frac{1}{\sqrt{{N}}}
            \sum_{t=0}^{{T}}
            U_tU_{t-1}...U_1 \ket{0^m}\otimes \ket{\tilde{w}}\otimes \ket{t},
        \end{equation*}
        where $\ket{\tilde{w}} = \ket{w_{\text{yes}}}$ for the YES case and $\ket{\tilde{w}} = \ket{w_{\text{no}}}$ for the NO case.
        Let the eigenstates of $H_0+H_1$ be $\ket{E_0^{0+1}}, \ket{E_1^{0+1}}, \dots$ with eigenvalues  $E_0^{0+1} < E_1^{0+1} \leq \dots$ and expand the history state as $\ket{\psi_{\text{hist,acc}}} = c_0\ket{E_0^{0+1}} + \sum_{k\geq1} c_k \ket{E_k^{0+1}}$.
        Then we obtain
        \begin{align*}
            \braket{\psi_{\text{hist,acc}}|H_0+H_1|\psi_{\text{hist,acc}}}
            & \geq E_0^{0+1}|c_0|^2 + E_1^{0+1}(1-|c_0|^2)      \\
            & \geq E_0^{0+1}|c_0|^2 + \left(E_0^{0+1} + \frac{\varepsilon}{{N}} - O(\varepsilon^2 / \Delta_0) \right)(1-|c_0|^2) \\
            & = E_0^{0+1} + \left( \frac{\varepsilon}{{N}} - O(\varepsilon^2 / \Delta_0) \right)(1-|c_0|^2).
        \end{align*}
        Now we have
        \begin{equation*}
            E_0^{0+1} \leq \braket{\psi_{\text{hist,acc}}|H_0+H_1|\psi_{\text{hist,acc}}} \leq \frac{\varepsilon \mathcal{O}(2^{-n})}{{N}}.
        \end{equation*}
        Combining the above two equations, for sufficiently small $\varepsilon = o(\Delta_0/{N})$, we get
        \begin{equation*}
\begin{alignedat}{2}
          &\;&  \frac{\varepsilon \mathcal{O}(2^{-n})}{{N}}
                     &\geq \left( \frac{\varepsilon}{{N}} - \mathcal{O}(\varepsilon^2 / \Delta_0) \right)(1-|c_0|^2) \\
         \Leftrightarrow&&
         \frac{\mathcal{O}(2^{-n})}{1-\mathcal{O}(\varepsilon {N}/ \Delta_0)}
                    & \geq   (1-|c_0|^2)\\
       \Leftrightarrow&&     |c_0|^2  &\geq 1 - \mathcal{O}\left( 2^{-n} \right).
        \end{alignedat}
        \end{equation*}
        Thus we obtain
        \begin{equation}
        \label{eq:hist_E_0^0+1}
            |\braket{\psi_{\text{hist,acc}}|E_0^{0+1}}|^2
            \geq 1 - \mathcal{O}\left( 2^{-n} \right),
        \end{equation}
        which concludes that the ground state of $H_0+H_1$ is sufficiently close to $\ket{\psi_{\text{hist,acc}}}$ when we have chosen $\varepsilon = o(\Delta_0/N) = o(1/T^4)$.

\proofpara{{Construction of the parameterized Hamiltonian}}
        Let $H_r(\lambda)$
        be
        $$
            H_r(\lambda) \coloneqq H_0 + H_1 + r V(\lambda)
        $$
        {where \(r\) is an inverse-polynomially small parameter chosen below}, and
        \[
            {
            V(\lambda):=
            \left(
            e^{2\pi i \lambda}\ket{1}\bra{0}_{\text{out}1}
            + e^{-2\pi i \lambda} \ket{0}\bra{1}_{\text{out}1}
            \right)
            \otimes \ket{1}\bra{1}_{C_T}
            }
        \]
        and therefore $V(\lambda)=V(\lambda)^\dagger$.

        {
        By amplification, there exist normalized states \(\ket{g_0}\) and \(\ket{g_1}\) such that}
        \[
        {
            U_x(\ket{0^m}\otimes\ket{\widetilde w})
            =
            \begin{cases}
                \ket{1}_{\mathrm{out}1}\ket{1}_{\mathrm{out}2}\ket{g_0}
                +\mathcal{O}(2^{-n}),
                & x\in L_{\mathrm{yes}},\\
                \ket{0}_{\mathrm{out}1}\ket{1}_{\mathrm{out}2}\ket{g_1}
                +\mathcal{O}(2^{-n}),
                & x\in L_{\mathrm{no}}.
            \end{cases}
        }
        \]
        {
        Here and below, \(\mathcal{O}(2^{-n})\) denotes a vector of exponentially small norm.
        }

\proofpara{{Analysis of parameterized ground states}}
        With second-order perturbation theory, the ground state of $H_r(\lambda)$ can be written as
        \begin{align*}
            \ket{\psi_r(\lambda)}= & \left( 1 -\frac{r^2}{2}
            \sum_{k\neq 0}\frac{|\bra{E_k^{0+1}}V(\lambda)\ket{E_0^{0+1}}|^2}{(E_k^{0+1} - E_0^{0+1})^2}\right)\ket{E_0^{0+1}}- r\sum_{k\neq 0} \frac{\bra{E_k^{0+1}}V(\lambda)\ket{E_0^{0+1}}}{E_k^{0+1} - E_0^{0+1}} \ket{E_k^{0+1}} \\
                                   & + r^2 \ket{\perp(\lambda)}
            + \mathcal{O}(r^3)
        \end{align*}
        where $E_k^{0+1}$ is the $k$-th eigenvalue of $H_0+H_1$, $\ket{E_k^{0+1}}$ is the corresponding eigenvector and
        \begin{equation*}
            \begin{split}
                \ket{\perp(\lambda)}
                \coloneqq
                 & \sum_{k\neq 0, l\neq 0}\frac{\bra{E_k^{0+1}}V(\lambda)\ket{E_l^{0+1}}\bra{E_l^{0+1}}V(\lambda)\ket{E_0^{0+1}}}{(E_k^{0+1} - E_0^{0+1})(E_l^{0+1} - E_0^{0+1})} \ket{E_k^{0+1}} \\
                 & -\sum_{k\neq 0}\frac{\bra{E_k^{0+1}}V(\lambda)\ket{E_0^{0+1}}\bra{E_0^{0+1}}V(\lambda)\ket{E_0^{0+1}}}{(E_k^{0+1} - E_0^{0+1})^2} \ket{E_k^{0+1}}.
            \end{split}
        \end{equation*}
\proofpara{{Analysis of the Berry connection and Berry phase}}
        Now,
        \begin{equation*}
            \begin{split}
                \frac{\mathrm{d}}{\mathrm{d}\lambda} \ket{\psi_r(\lambda)}
                = & -r^2 \sum_{k\neq 0}\frac{\mathrm{Re} \left(\bra{E_0^{0+1}}V(\lambda)\ket{E_k^{0+1}}\bra{E_k^{0+1}}\frac{\mathrm{d} V(\lambda)}{\mathrm{d} \lambda}\ket{E_0^{0+1}}\right)}{(E_k^{0+1} - E_0^{0+1})^2}\ket{E_0^{0+1}}
                \\
                  & - r\sum_{k\neq 0} \frac{\bra{E_k^{0+1}}\frac{\mathrm{d} V(\lambda)}{\mathrm{d} \lambda}\ket{E_0^{0+1}}}{E_k^{0+1} - E_0^{0+1}} \ket{E_k^{0+1}}
                + r^2 \frac{\mathrm{d}}{\mathrm{d}\lambda} \ket{\perp(\lambda)}
                + \mathcal{O}(r^3).
            \end{split}
        \end{equation*}
        {For simplicity, define}
        \[
            {
            R_0^2\coloneqq
            \sum_{k\neq 0}
            \frac{\ket{E_k^{0+1}}\bra{E_k^{0+1}}}
            {(E_k^{0+1}-E_0^{0+1})^2}.
            }
        \]
        Observing that
        \begin{equation*}
            \bra{E_0^{0+1}} \frac{\mathrm{d}}{\mathrm{d}\lambda} \ket{\perp(\lambda)} = 0,
        \end{equation*}
        we obtain
        \begin{align*}
            \bra{\psi_r(\lambda)}\frac{\mathrm{d}}{\mathrm{d}\lambda} \ket{\psi_r(\lambda)}
            = &
                -r^2 \sum_{k\neq 0}\frac{1}{(E_k^{0+1} - E_0^{0+1})^2}
                \mathrm{Re}\!\left(
                \begin{aligned}
                &\bra{E_0^{0+1}}V(\lambda)\ket{E_k^{0+1}}\\[-0.2em]
                &\quad\times\bra{E_k^{0+1}}\frac{\mathrm{d} V(\lambda)}{\mathrm{d} \lambda}\ket{E_0^{0+1}}
                \end{aligned}
                \right) \\*
                  & + r^2 \sum_{k\neq 0,k'\neq 0}
                \begin{aligned}[t]
                &\frac{\bra{E_0^{0+1}}V(\lambda)\ket{E_k^{0+1}}}{E_k^{0+1} - E_0^{0+1}}\\[-0.2em]
                &\quad\times
                \frac{\bra{E_{k'}^{0+1}}\frac{\mathrm{d} V(\lambda)}{\mathrm{d} \lambda}\ket{E_0^{0+1}}}{E_{k'}^{0+1} - E_0^{0+1}}
                \braket{E_k^{0+1}|E_{k'}^{0+1}}
                \end{aligned}
                +\mathcal{O}(r^3) \\*
            = &
            i r^2
            \mathrm{Im}\left( \bra{E_0^{0+1}}
            V(\lambda)
            {R_0^2}
            \frac{\mathrm{d} V(\lambda)}{
                \mathrm{d} \lambda}\ket{E_0^{0+1}} \right)
            +\mathcal{O}(r^3).
        \end{align*}

        {Choosing the global phase of \(\ket{E_0^{0+1}}\) appropriately, Eq.~\eqref{eq:hist_E_0^0+1} gives}
        \[
            {
            \left\|\ket{E_0^{0+1}}-\ket{\psi_{\mathrm{hist,acc}}}\right\|
            =\mathcal{O}(2^{-n}).
            }
        \]
        {Define the normalized states}
        \[
            {
            \ket{\phi_{\mathrm{yes}}}
            \coloneqq
            \ket{0}_{\mathrm{out}1}\ket{1}_{\mathrm{out}2}\ket{g_0}\ket{T},
            \qquad
            \ket{\phi_{\mathrm{no}}}
            \coloneqq
            \ket{1}_{\mathrm{out}1}\ket{1}_{\mathrm{out}2}\ket{g_1}\ket{T}.
            }
        \]
        {Since}
        \[
            {
            \frac{\mathrm{d}V(\lambda)}{\mathrm{d}\lambda}
            =
            2\pi i
            \left(
            e^{2\pi i\lambda}\ket{1}\bra{0}_{\mathrm{out}1}
            -
            e^{-2\pi i\lambda}\ket{0}\bra{1}_{\mathrm{out}1}
            \right)
            \otimes\ket{1}\bra{1}_{C_T},
            }
        \]
        {the final-time term of the history state gives}
        \begin{align*}
            \color{black}
            V(\lambda)\ket{E_0^{0+1}}
            &=
            \begin{cases}
                \dfrac{e^{-2\pi i\lambda}}{\sqrt N}\ket{\phi_{\mathrm{yes}}}
                +\mathcal{O}(2^{-n}),
                & x\in L_{\mathrm{yes}},\\[0.4em]
                \dfrac{e^{2\pi i\lambda}}{\sqrt N}\ket{\phi_{\mathrm{no}}}
                +\mathcal{O}(2^{-n}),
                & x\in L_{\mathrm{no}},
            \end{cases}\\
            \frac{\mathrm{d}V(\lambda)}{\mathrm{d}\lambda}\ket{E_0^{0+1}}
            &=
            \begin{cases}
                -\dfrac{2\pi i e^{-2\pi i\lambda}}{\sqrt N}\ket{\phi_{\mathrm{yes}}}
                +\mathcal{O}(2^{-n}),
                & x\in L_{\mathrm{yes}},\\[0.4em]
                \dfrac{2\pi i e^{2\pi i\lambda}}{\sqrt N}\ket{\phi_{\mathrm{no}}}
                +\mathcal{O}(2^{-n}),
                & x\in L_{\mathrm{no}}.
            \end{cases}
        \end{align*}
        {Omitting the exponentially small terms, we obtain}
        \[
            {
            i\mathcal{A}_\lambda=
            \begin{cases}
                \dfrac{2\pi r^2}{N}
                \bra{\phi_{\mathrm{yes}}}R_0^2\ket{\phi_{\mathrm{yes}}}
                +\mathcal{O}(r^3),
                &x\in L_{\mathrm{yes}},\\[0.6em]
                -\dfrac{2\pi r^2}{N}
                \bra{\phi_{\mathrm{no}}}R_0^2\ket{\phi_{\mathrm{no}}}
                +\mathcal{O}(r^3),
                &x\in L_{\mathrm{no}},
            \end{cases}
            }
        \]
        where $\mathcal{A}_\lambda=\bra{\psi_r(\lambda)}\frac{\mathrm{d}}{\mathrm{d}\lambda} \ket{\psi_r(\lambda)}$.
        Observe that $\Delta_{0+1} \leq (E_k^{0+1} - E_0^{0+1}) \leq K\coloneqq \|H_0+H_1\|= {\mathcal{O}(N+n)}$ and the following holds:
            \[
                0 \preceq
                \frac{I-\ket{E_0^{0+1}}\bra{E_0^{0+1}}}{K^2}
                \preceq R_0^2
                \preceq
                \frac{I-\ket{E_0^{0+1}}\bra{E_0^{0+1}}}{\Delta_{0+1}^2}.
            \]
        This implies for any state $\ket{\phi}\in \mathcal{H}_w \otimes \mathcal{H}_m\otimes\mathcal{H}_{\mathrm{clock}}$,
            \begin{align*}
            0\leq
                &\bra{\phi} \frac{I-\ket{E_0^{0+1}}\bra{E_0^{0+1}}}{K^2} \ket{\phi} \\*
                \leq{}& \bra{\phi}R_0^2\ket{\phi} \\*
                \leq{}& \bra{\phi} \frac{I-\ket{E_0^{0+1}}\bra{E_0^{0+1}}}{\Delta_{0+1}^2} \ket{\phi}.
            \end{align*}
        {We apply this inequality to \(\ket{\phi_{\mathrm{yes}}}\) and \(\ket{\phi_{\mathrm{no}}}\) in the YES and NO cases, respectively. Since these states have the opposite value of the first output qubit from the corresponding final terms of the history states,}
        \[
            {
            |\braket{E_0^{0+1}|\phi_{\mathrm{yes}}}|
            =\mathcal{O}(2^{-n})
            }
        \]
        {for \(x\in L_{\mathrm{yes}}\), and}
        \[
            {
            |\braket{E_0^{0+1}|\phi_{\mathrm{no}}}|
            =\mathcal{O}(2^{-n})
            }
        \]
        {for \(x\in L_{\mathrm{no}}\).}
        Therefore, if $x\in L_{\text{yes}}$,
        \begin{align*}
            \color{black}{
            0< 2\pi \frac{r^2}{N K^2}
            \left(1-\mathcal{O}(2^{-n})\right)
            -\mathcal{O}(r^3)
            \leq i\mathcal{A}_\lambda
            \leq  2\pi \frac{r^2}{N\Delta_{0+1}^2}
            +\mathcal{O}(r^3)}
        \end{align*}
        and if $x\in L_{\text{no}}$,
        \begin{align*}
            \color{black}{
            0
            <2\pi \frac{r^2}{N K^2}
            \left(1-\mathcal{O}(2^{-n})\right)
            -\mathcal{O}(r^3)
            \leq -i\mathcal{A}_\lambda
            \leq  2\pi \frac{r^2}{N\Delta_{0+1}^2}
            +\mathcal{O}(r^3),}
        \end{align*}
        where $\mathcal{A}_\lambda=\bra{\psi(\lambda)}\frac{\mathrm{d}}{\mathrm{d}\lambda}\ket{\psi(\lambda)}$.

        {
        We now specify the parameter choices.
        We choose an \(\varepsilon=o(\Delta_0/N)=o(1/N^4)\), as required in the analysis of \(H_0+H_1\).
        Then we choose a sufficiently small \(r\) satisfying
        }
        \[
            {
            r\leq
            \min\left\{
                \frac{\Delta_{0+1}}{4},
                O\left(\frac{1}{N K^2}\right)
            \right\}
            \leq
            \min\left\{
                O(\varepsilon/N),
                O\left(\frac{1}{N (N+n)^2}\right)
            \right\}.
            }
        \]
        {
        The first condition preserves the spectral gap.
            The second is needed so that
        \[
            2\pi \frac{r^2}{N K^2}
            \left(1-\mathcal{O}(2^{-n})\right)
            -\mathcal{O}(r^3)\geq \frac{1}{\mathrm{poly}(n)}.
        \]
        Then we may take $\delta \geq 1/\text{poly}(n)$
        }
        {and obtain}
        $$
            \delta \leq i\mathcal{A}_\lambda\leq \frac{\pi}{2}
        $$
        if $x\in L_{\text{yes}}$ and
        $$
            \delta \leq -i\mathcal{A}_\lambda\leq \frac{\pi}{2}
        $$
        if $x\in L_{\text{no}}$ for some $\delta >1/\mathrm{poly}(n)$.

        The Berry phase is given by
        $$
            \theta_B= i\int_0^1 \mathcal{A}_\lambda\, \mathrm{d}\lambda.
        $$
        Therefore,
        \begin{equation*}
            \theta_B\in
            \begin{cases}
                \left[\delta, \frac{\pi}{2}\right]       & \text{if } x\in L_\text{yes} \\
                \left[\frac{3\pi}{2}, 2\pi-\delta\right] & \text{if } x\in L_\text{no}.
            \end{cases}
        \end{equation*}
        Corresponding to the definition of the problem, by taking $a=0$ and $b=\pi$, $\theta_B \in[a+\delta/2,b-\delta/2]_{2\pi}$ if $x\in\mathcal{L}_{\mathrm{yes}}$ and $\theta_B\in[b+\delta/2,a-\delta/2]_{2\pi}$ if $x\in\mathcal{L}_{\mathrm{no}}$.

\proofpara{{Spectral gap and energy threshold}}
        Recall
        $
            H_r(\lambda) = H_0 + H_1 + r V(\lambda)
        $
        and the spectral gap of $H_0 + H_1$ is larger than $\Delta_{0+1} \geq \frac{\varepsilon}{N} - O(\varepsilon^2 / \Delta_0) \geq 1/\mathrm{poly}(n)$.
        Because $\|r V(\lambda)\|=r$, by taking $r\leq \Delta_{0+1}/4$,
        the spectral gap of $H_r(\lambda) = H_0 + H_1 + r V(\lambda)$ is lower bounded by $\Delta_{0+1}/2$ for any $\lambda$.
        We can also consider the threshold $E_\text{th}$ for the ground state energy of $H_r(\lambda) = H_0 + H_1 + r V(\lambda)$.
        In both YES and NO case, the ground state energy is $E_0 \leq E_0^{0+1} + O(r)\leq \frac{\varepsilon \mathcal{O}(2^{-n})}{N}+O(r)$ and the first excited state energy is $E_1 \geq E_1^{0+1} -O(r) \geq \frac{\varepsilon}{N} - O(\varepsilon^2 / \Delta_0) -O(r)$.
        Therefore we can choose $E_\text{th}=\frac{\varepsilon}{2N}$ for sufficiently small $\varepsilon$ and $r$.
    \end{proof}

    \subsection{Proof of Corollary~\ref{thm:P^PGQMAcomp}: Containment in $\mathsf{P}^{\mathsf{PGQMA[log]}}$
    }
    \label{sec:P^PGQMAcomp}
    \label{sec:P^PGQMA_containment}
    \begin{proof}
    Now we show that BPE without energy threshold can be solved with $O(\log{(n)})$ queries to a $\mathsf{PGQMA}$ oracle.
    Our strategy is mainly based on Ref.~\cite{ambainis2014physical}.
    The algorithm consists of the following two steps:
    \begin{enumerate}
        \item Query the $\mathsf{PGQMA}$ oracle $O(\log{(n)})$ times and perform binary search to obtain an estimate $\hat{E}_0$ of the ground state energy $E_0$ such that $\hat{E}_0 \in [E_0, E_0 + \Delta_\mathrm{min}/2]$, where $\Delta_\mathrm{min}$ is a lower-bound of the spectral gap of a given Hamiltonian $H(\lambda)$ for all $\lambda$.
        This can be done since the problem of estimating the ground state energy of Hamiltonians with inverse polynomial gap is \textsf{PGQMA}-complete, which was shown in Ref.~\cite{aharonov2022pursuit}.
        \item Solve BPE with energy threshold {$\hat{E}_0 + \Delta_\mathrm{min}/4$} using a single query to the {$\UQMA \cap \mathsf{co}\text{-}\UQMA$} oracle.
              {This threshold satisfies the promise in Definition~\ref{def:BPE_withEth}, since it lies between \(E_0+\Delta_{\mathrm{min}}/4\) and \(E_1-\Delta_{\mathrm{min}}/4\).}
              {Note that this step can be done with the $\mathsf{PGQMA}$ oracle since $\UQMA \cap \mathsf{co}\text{-}\UQMA \subseteq \mathsf{PGQMA}$.}
    \end{enumerate}
    \end{proof}

\subsubsection{$\mathsf{P}^{{\UQMA \cap \mathsf{co}\text{-}\UQMA}\mathsf{[log]}}$-hardness under polynomial-time truth-table reductions}
    \label{sec:P^dUQMA_hardness}

    We mainly follow the proof of $\mathsf{P}^{\mathsf{QMA[log]}}$-hardness under polynomial-time truth-table reductions of the Spectral Gap problem in Ref.~\cite{yirka2025note}.
    The following is the corollary of Theorem~\ref{thm:UQMAcomp}.

    \begin{corollary}
        \label{cor:PGQMAhard}
        Berry phase estimation problem without energy threshold is {$\UQMA \cap \mathsf{co}\text{-}\UQMA$}-hard under polynomial-time many-one/Karp reductions.
    \end{corollary}

    Now we are ready to prove the $\mathsf{P}^{{\UQMA \cap \mathsf{co}\text{-}\UQMA}\mathsf{[log]}}$-hardness.

    \begin{proof}[Proof of the $\mathsf{P}^{{\mathsf{UQMA} \cap \mathsf{co}\text{-}\mathsf{UQMA}}\mathsf{[log]}}$-hardness in Corollary~\ref{thm:P^PGQMAcomp}]
        First, we observe that BPE without energy threshold is $\mathsf{P}^{{\UQMA \cap \mathsf{co}\text{-}\UQMA}\mathsf{[log]}}$-hard under polynomial-time Turing reductions because any $\mathsf{P}^{{\UQMA \cap \mathsf{co}\text{-}\UQMA}\mathsf{[log]}}$ computation can be simulated in polynomial-time by substituting any query to {$\UQMA \cap \mathsf{co}\text{-}\UQMA$}~ oracle with a query to ``BPE without energy threshold''.
        Corollary~\ref{cor:PGQMAhard} ensures that one query to {$\UQMA \cap \mathsf{co}\text{-}\UQMA$}~ oracle can be substituted by one query to ``BPE without energy threshold''.
        Thus, the above polynomial-time Turing reduction requires at most a logarithmic number of queries.
        As shown in Refs.~\cite{beigel1991bounded, yirka2025note}, a logarithmic number of adaptive queries can be simulated by a polynomial number of non-adaptive queries.
        This implies that BPE without energy threshold is $\mathsf{P}^{{\mathsf{UQMA}\cap\mathsf{co}\text{-}\mathsf{UQMA}}\mathsf{[log]}}$-hard under polynomial-time truth-table reductions.
    \end{proof}

    \section*{Author contributions}
R.H. and K.S. contributed equally to this work. All authors contributed to the research, discussed the results, and reviewed and approved the final manuscript. No artificial intelligence model was used in preparing version 1 of the arXiv manuscript. ChatGPT 5.6 Sol, accessed via Codex, was used to assist with revisions made after that version, including language editing, formatting, consistency checks, and checks of mathematical proofs. All mathematical arguments and scientific claims were verified by the authors, who take full responsibility for the manuscript.

    \section*{Acknowledgements}
RH was supported by JST PRESTO Grant Number JPMJPR23F9 and MEXT KAKENHI Grant Number 21H05183.\par
KS was supported by MEXT Q-LEAP Grant No. JPMXS0120319794 and JST SPRING Grant No. JPMJSP2138.\par
CK was supported by JST ASPIRE Japan Grant Number JPMJAP2319, JST PRESTO, Grant Number JPMJPR25F1 and JSPS KAKENHI, Grant Number JP26K17052.
    \bibliographystyle{quantum}
    \bibliography{main}

    \appendix

    \section{Review of BPE Algorithm by Murta et al.\label{sec:berry_algorithms}}

    Here, we review an earlier algorithm introduced by Murta et al.~\cite{murta2020berry}, which estimates the Berry phase $\theta_B$ acquired by the non-degenerate ground state of a $k$-local Hamiltonian family $\{H(\lambda)\}_{\lambda \in [0,1]}$ undergoing cyclic adiabatic evolution. The strategy involves simulating the closed loop twice—once forward and once backward—so that the dynamical phases cancel and the geometric (Berry) phases accumulate. However, this method relies on time-reversal symmetry and therefore restricts the Berry phase $\theta_B$ to the interval $[0, \pi)$.
In this section, we outline Murta et al.'s algorithm in detail. Our improved algorithm, which removes the symmetry requirement and allows estimation of \( \theta_B \) over the full range \( [0, 2\pi) \), is presented separately in Section~\ref{sec:improved_algorithm}.

To isolate the Berry phase \( \theta_B \), Murta et al.\ consider evolving the system along the closed loop with the Hamiltonian sign reversed:
\[
    \bar U(T) = \mathcal{T} \exp\left(+i \int_0^T H(\lambda(t)) dt\right).
\]
Since the eigenstates of \( H \) remain unchanged, the dynamical phase flips sign while the Berry phase does not. Consequently,
\[
    \bar U(T) U(T) |\psi(0)\rangle = e^{i 2\theta_B} |\psi(0)\rangle,
\]
so the composite evolution isolates \( \theta_B \) while canceling \( \theta_D \). The algorithm proceeds as follows:

\begin{enumerate}
    \item \textbf{Ground-state preparation.} Prepare the ground state \( \ket{\psi(0)} \) of \( H(0) \).

    \item \textbf{Digitized adiabatic evolution.} Discretize the loop into \( N = \operatorname{poly}(n,1/\varepsilon) \) steps of duration \( \delta t \) satisfying \( \delta\lambda/\delta t \ll \Delta^2 \) to ensure adiabaticity. The resulting unitaries are
          \[
              U_{\text{loop}}(T,0) = \prod_{j=1}^{N} e^{-i H(\lambda_j)\delta t}, \quad
              \bar{U}_{\text{loop}}(T,0) = \prod_{j=1}^{N} e^{+i H(\lambda_j)\delta t},
          \]
          which can be implemented using a Trotter-Suzuki decomposition with polynomial depth.

    \item \textbf{Dynamical-phase cancellation.} Apply the time-reversed evolution \( \bar U_{\text{loop}} \) immediately after \( U_{\text{loop}} \). On the ground state, this yields
          \[
              \bar U_{\text{loop}} U_{\text{loop}} \ket{\psi(0)} = e^{i 2\theta_B} \ket{\psi(0)}.
          \]
          (For time-reversal-symmetric families, it suffices to reverse the sign of the Hamiltonian during the second half of the loop, reducing circuit depth by half.)

    \item \textbf{Phase Estimation.} Use standard quantum phase estimation (QPE) on the unitary \( \bar U_{\text{loop}} U_{\text{loop}} \), controlled on an \( m = \lceil \log_2(1/\varepsilon) \rceil \)-qubit register. This yields an estimate \( p \) such that \( |2\theta_B / 2\pi - p / 2^m| < \varepsilon \). In cases where \( \theta_B \) is known to be small (e.g., on finely discretized lattices), a Hadamard test with a single ancilla may be used instead, at the cost of increased sampling noise.
\end{enumerate}

The adiabatic evolution subroutine requires a total runtime \( T = \operatorname{poly}(n, 1/\Delta) \), and the full loop is discretized into \( N = \operatorname{poly}(n, 1/\varepsilon) \) Trotter steps. Phase estimation requires \( O(\log(1/\varepsilon)) \) controlled applications of \( \bar{U}_{\text{loop}} U_{\text{loop}} \), each with circuit depth polynomial in \( n \), assuming the Hamiltonian is \( k \)-local and efficiently simulable. Therefore, the entire procedure runs in \textsf{BQP}-time and estimates \( \theta_B \) to inverse-polynomial precision.

However, due to the phase doubling \( e^{i2\theta_B} \), this method restricts the range of estimable Berry phases to \( [0,\pi) \) in the absence of time-reversal symmetry. Our improved algorithm, introduced in Section~\ref{sec:improved_algorithm}, overcomes this limitation and enables full-range estimation modulo \( 2\pi \) without requiring symmetry assumptions.

\end{document}